\documentclass{rsproca_new}
\usepackage{bm}
\usepackage{dsfont}
\usepackage{textcomp} 
\usepackage{todonotes}
\newtheorem{cor}{Corollary}
\newtheorem{lemma}{Lemma}
\newtheorem{rmk}{Remark}
\newcommand{\assign}{:=}

\newcommand{\dif}{\mathrm{d}}
\newcommand{\E}{\mathbb{E}}
\newcommand{\1}{\mathds{1}}
\newcommand{\R}{\mathbb{R}}
\newcommand{\PP}{\mathbb{P}}
\newcommand{\bmalpha}{\bm{\alpha}}
\newcommand{\eps}{\varepsilon}
\newcommand{\ra}{\rangle}
\newcommand{\la}{\langle}

\newcommand{\sGK}{\sigma^2}

\newcommand{\Ord}{\mathcal{O}}

\newcommand{\ds}{(\ref{eq.skewx})-(\ref{eq.skewy})}
\newcommand{\xxi}{{\rm{x}}}

\newcommand{\highlight}[1]{{#1}}

\begin{document}
	
	\title{Edgeworth expansions for slow-fast systems with finite time scale separation}
		
	\author{
		Jeroen Wouters$^{1,2}$
		Georg A.~Gottwald$^{3}$
	}
	
	\address{$^{1}$Department of Mathematics and Statistics, University of Reading, Reading, United Kingdom\\
		$^{2}$Niels Bohr Institute, University of Copenhagen, Copenhagen, Denmark\\
		$^{3}$School of Mathematics and Statistics, University of Sydney, NSW 2006, Australia}
	
	\subject{70K70, 65C20, 37A50, 60F05}
	
	
	\corres{Jeroen Wouters\\
		\email{j.wouters@reading.ac.uk}}
	
\begin{abstract}
	We derive Edgeworth expansions that describe corrections to the Gaussian limiting behaviour of slow-fast systems. The Edgeworth expansion is achieved using a semi-group formalism for the transfer operator, where a Duhamel-Dyson series is used to asymptotically determine the corrections at any desired order of the time scale parameter $\eps$. The corrections involve integrals over higher-order auto-correlation functions. We develop a diagrammatic representation of the series to control the combinatorial wealth of the asymptotic expansion in $\eps$ and provide explicit expressions for the first two orders. At a formal level, the expressions derived are valid in the case when the fast dynamics is stochastic as well as when the fast dynamics is entirely deterministic. We corroborate our analytical results with numerical simulations and show that our method provides an improvement on the classical homogenization limit which is restricted to the limit of infinite time scale separation. 
\end{abstract}
	
\begin{fmtext}
\section{Introduction}
Many systems in the natural sciences feature a time scale separation between slowly and rapidly evolving variables. Examples range from molecular drug design where a protein interacts with a target molecule in the presence of a rapidly fluctuating environment \cite{kamerlin_coarse-grained_2011}, to climate dynamics where the slowly evolving ocean dynamics is driven by rapidly evolving atmospheric weather systems \cite{imkeller_stochastic_2001}. 

Such systems can often be modeled by multi-scale systems of the form
\end{fmtext}
\maketitle


\begin{align}
  \dif{x} & =  \frac{1}{\varepsilon} f_0 (x, y) \dif t + f_1(x, y) \dif t
   \label{eq.skewx}
  \\
  \dif{y} & =  \frac{1}{\varepsilon^2} g_0 (y) \dif t +  \frac{1}{\varepsilon} \beta(y) \dif W_t+  \frac{1}{\varepsilon} g_1 (x,y)  \dif t
 \label{eq.skewy}
 .
\end{align}
with $x\in \R^d$, $y\in \R^m$ and $\beta \in \R^{m\times l}$ and $l$-dimensional Wiener process $W_t$. The fast dynamics can be stochastic with $\beta\neq 0$ or deterministic with $\beta\equiv 0$. In the stochastic case we assume that the fast dynamics $\dif y = g_0 \dif t + \beta \dif W_t $ is an ergodic process with an absolutely continuous measure $\mu$ and the full system \eqref{eq.skewx}-\eqref{eq.skewy} admits an absolutely continuous measure $\mu^{(\eps)}$. In the purely deterministic case $\beta\equiv 0$ we assume that the fast dynamics $\dif y = g_0 \dif t$ admits a unique invariant physical measure $\mu$ on $\R^m$ and the full system \eqref{eq.skewx}-\eqref{eq.skewy} admits a unique invariant physical measure $\mu^{(\eps)}$ on $\R^{d+m}$ \footnote{An ergodic measure is called physical if for a set of initial conditions of nonzero Lebesgue measure the temporal average of a continuous observable converges to the spatial average over this measure.}. Here $\eps\ll 1$ denotes the degree of time scale separation between the slow and fast variables, $x$ and $y$, respectively. Often only the slow variables are of interest in such systems, and one seeks reduced equations for the slow dynamics only. 

In the limit of infinite time scale separation $\varepsilon \to 0$, closed diffusion equations for the slow variables can be obtained by the method of homogenization \cite{GivonEtAl04,Pavliotis08,Khasminsky66,Kurtz73,Papanicolaou76,Beck90,JustEtAl01}. The diffusive behaviour emerges as the integrated effect of the fast dynamics, reminiscent of the summation of random variables in the central limit theorem (CLT). This method applies to slow-fast systems, where the fast dynamics can be either stochastic or deterministic with $\beta\equiv 0$. In the deterministic case the fast dynamics needs to be sufficiently chaotic. If the leading order slow vector field averages to zero, i.e. if $\int f_0(x,y)\mu(\dif y) = 0$, the slow dynamics is approximated on time scales of $\Ord(1)$ by a stochastic differential equation; see \cite{Khasminsky66,Kurtz73,Papanicolaou76} for the stochastic case and \cite{MelbourneStuart11,GottwaldMelbourne13c,KellyMelbourne17} and \cite{Dolgopyat04,DeSimoiLiverani15,DeSimoiLiverani16} for the deterministic case. Homogenization has been used to design efficient numerical multi-scale integrators such as {equation-free} projection \cite{GearKevrekidis03,KevrekidisGearEtAl03} and the {heterogeneous multi-scale method} \cite{E03,EEtAl07}, and has been used for stochastic parameterization in the climate sciences \cite{MTV99,Majdaetal01,MTV02,Majda03,monahan_stochastic_2011,culina_stochastic_2011,GottwaldEtAl17}.\\


In realistic physical systems, however, the time scale separation is always finite. In systems without a clear time scale separation classical homogenization theory may fail and may not be able to reliably approximate the stochastic long-time behaviour of the slow dynamics. For finite values of $\varepsilon$ there is an intricate feedback between the evolution of the slow $x$ variables and the fast $y$ variables which prevents the approximation of the integrated fast dynamics by Brownian motion. Similar issues arise in the CLT, where the Gaussian distribution is only an accurate approximation for sums with a sufficiently large number of summands. For the CLT techniques exist to obtain a more accurate description of the distribution of finite sums than provided by the limiting Gaussian. A classical technique is the Edgeworth expansion, which provides an expansion of the distributions of sums, asymptotic in $1/\sqrt{n}$, where $n$ is the length of the sum \cite{BhattacharyaRanga,fellerbook}.\\ 

In the case of a multi-scale system such as (\ref{eq.skewx})-(\ref{eq.skewy}) the small parameter controlling the limit is now $\eps$, instead of $1/\sqrt{n}$, and the random variables converging to a Gaussian are the increments over time of the slow variable $x$. The multi-scale system (\ref{eq.skewx})-(\ref{eq.skewy}) features three distinct time scales \cite{simoi_fastslow_2017}: the fast time scale of $\Ord(\eps^2)$, an intermediate time scale of $\Ord(\eps)$ on which the slow dynamics is trivial but the fast dynamics has equilibrated, and the long diffusive time scale of $\Ord(1)$ on which the slow dynamics exhibits nontrivial diffusive behaviour. The corrections to Gaussianity occur on the intermediate time scale; it is sufficiently long for the integrated noise on $x$ to become nearly Gaussian, but not long enough for the slow dynamics to dominate. To focus on the statistical behaviour on the intermediate time scale we consider the transition probabilities of the slow variable $x$
\begin{equation*}
\pi_\eps(\xxi,t,x_0) = \PP \left( \left. \frac{x(t) - x(0) }{\sqrt{t}} \in (\xxi, \xxi + \mathrm{d} \xxi) \right| x(0) = x_0, \, y(0) \sim   \mu_{x_0}^{(\eps)}\right) ,
\end{equation*}
where $\mu_{x_0}^{(\eps)}$ is the invariant measure of  (\ref{eq.skewx})-(\ref{eq.skewy}) conditioned on $x=x_0$ and $x(t)$ the slow variable of a solution of the multi-scale system \ds.

Homogenization dictates that on the intermediate time scale and in the limit  $\eps \rightarrow 0$ of infinite time scale separation, $\pi_{\eps}(t=\eps)$ becomes Gaussian. We will refer to this convergence as the CLT in the context of slow-fast systems.
 For small but finite $\varepsilon$, the deviations from Gaussianity of $\pi_\eps$ will be small. We can therefore expand $\pi_{\eps} ( \xxi, \eps, x_0 ) $ in $\sqrt{\eps}$. This expansion is the equivalent of the classical Edgeworth expansion for slow-fast systems.

Whereas the limiting Gaussian probability, implied by homogenization theory, only involves the two-time statistics, higher-order Edgeworth expansions involve higher-order time correlations, containing more information about the dynamics. This constitutes our main result:
\subsection{Main result}
We assume that the fast dynamics of the multi-scale system (\ref{eq.skewx})-(\ref{eq.skewy}) decorrelates sufficiently rapidly, such that higher-order correlation functions are integrable and all expectation values appearing in the formulae below exist. Furthermore we assume that the conditional invariant measure $\mu^{(\eps)}_{x_0}$ obeys linear response w.r.t $\eps$, and $\int f_0 \, \mu (dy)=0$. 
Then the expansion is given up to ${\mathcal{O}}(\eps^\frac{3}{2})$ in the limit $\eps \rightarrow 0$ with $t= \eps\ll 1$ by 
\begin{align}
  &\pi_\eps(\xxi,t= \eps,x_0) = \mathbf{n}_{0,\sigma^2} (\xxi) \left( 1 
        + \sqrt{\eps} \left( \frac{c^{(1)}_{\frac{1}{2}}}{\sigma} H_1\left(\frac{\xxi}{\sigma}\right) + \frac{c^{(3)}_{\frac{1}{2}}}{3! \sigma^3} H_3\left(\frac{\xxi}{\sigma}\right) \right) \right. 
        \label{e.gamma_xi}  \\
         &\qquad + \left. \eps \left(\frac{ {c^{(2)}_{1}} + {{c^{(1)}_{\frac{1}{2}}}^2}}{2 \sigma^2}  H_2\left(\frac{\xxi}{\sigma}\right) + \frac{c^{(4)}_{1} + 4 {c^{(1)}_{\frac{1}{2}} c^{(3)}_{\frac{1}{2}}}}{4! \sigma^4}  H_4\left(\frac{\xxi}{\sigma}\right) + \frac{{c^{(3)}_{\frac{1}{2}}}^2}{2 (3! \sigma^3)^2} H_6\left(\frac{\xxi}{\sigma}\right) \right)
        \right) \nonumber\\ &\hspace{10cm}  + \Ord(\eps^{\frac{3}{2}}),
        \nonumber     
\end{align}
	with $c^{(p)}_{\frac{1}{2}} = c^{(p)}_{0,\frac{1}{2}} + c^{(p)}_{1,-\frac{1}{2}}$ and $c^{(p)}_1 = c^{(p)}_{0,1} + c^{(p)}_{1,0} + c^{(p)}_{2,-1} $,  with expressions given in equations (\ref{e.m1b}), (\ref{e.m2c}), (\ref{e.m2a}), (\ref{e.m2b}), (\ref{e.c3a}), (\ref{e.c3b}), (\ref{c4a}), (\ref{c4b}) and (\ref{c4c}), and where $H_n(\xxi) = (\xxi - \frac{\dif}{\dif \xxi})^n 1$ are Hermite polynomials of degree $n$. \\

Establishing this expansion involves expanding the first four cumulants $c^{(p)}$ of $(x(t)-x(0))/\sqrt{t}$ with $p\le 4$, resulting in their expansion coefficients $c^{(p)}_k$ with $k \in \lbrace \frac{1}{2}, 1 \rbrace$ (see Section \ref{sec.Econt}). These coefficients only involve the leading order measure $\mu_{x_0}^{(0)} = \mu$ and, in particular, do not involve the linear response correction to $\mu^{(\eps)}_{x_0}$. \highlight{It should be noted that the expressions for the cumulant expansion coefficients $c^{(p)}_{(i,j)}$ as derived below determine the functional form of the expansion, but are not sufficient to show that an Edgeworth expansion holds for a given class of dynamical systems.}\\

\subsection{Plan of the paper}
The paper is organized as follows. In Section~\ref{sec.homo} we briefly review homogenization theory highlighting the r\^ole of infinite time scale separation in convergence of transition probabilities to Gaussian distributions. Section~\ref{sec.edge} reviews Edgeworth expansions for discrete stochastic systems. Section~\ref{sec.Econt} introduces Edgeworth expansions and finite time scale separation corrections to the CLT for dynamical multi-scale systems and presents the explicit expression of the Edgeworth expansion; the lengthy and involved derivation is provided in the Appendix \ref{appendix}. In Section~\ref{sec.numerics} we present a numerical example corroborating our analytical main result for the transition probability and show that the Edgeworth expansion improves on the CLT. We conclude with a summary and an outlook in Section~\ref{sec.summary}. 


\section{Homogenization}
\label{sec.homo}
Homogenization describes the integrated effect of fast (either stochastic or deterministic chaotic) dynamics on slow variables as noise. Initially developed for stochastic multi-scale systems \cite{Khasminsky66,Kurtz73,Papanicolaou76}, homogenization has been extended recently to deterministic multi-scale systems when the fast dynamics evolves on a compact attractor $\Lambda \subset \R^m$ with an ergodic invariant measure $\mu$.  It was shown rigorously that for sufficiently chaotic fast dynamics the emergent stochastic long time behaviour of the slow dynamics is given by stochastic differential equations driven by Brownian motion \cite{Dolgopyat05,MelbourneStuart11,GottwaldMelbourne13c,KellyMelbourne17}. While homogenization for stochastic systems has been proven in a very general setting, the results in the deterministic case are so far limited to the skew product case with $g_1=0$. The assumptions on the chaoticity of the fast subsystem are mild, including Axiom A diffeomorphisms and flows, H\'enon-like attractors and Lorenz attractors.  
%

The following heuristic argument serves to show how the diffusive behaviour of the slow dynamics is linked to the CLT for time-integrated stationary processes in the case when the dynamical system is entirely deterministic. Consider the simplified version of the multi-scale system (\ref{eq.skewx})-(\ref{eq.skewy}) with $x\in \R$ and $y\in\R^m$
\begin{align*}
\dot x &= \frac{1}{\varepsilon}f_0(y)\\
\dot y &= \frac{1}{\varepsilon^2}g_0(y).
\end{align*}
Integrating the slow dynamics leads to
\begin{equation*}
x(t) = x(0) + \frac{1}{\varepsilon}\int_0^t f_0(y(s))\,\dif s.
\end{equation*}
Transforming the integrand to the fast time scale $\tau = s/\varepsilon^2$ we obtain 
\begin{equation*}
x(t) = x(0) + W_\eps(t),
\end{equation*}
where we introduced
\begin{equation}
W_\eps(t)=\varepsilon\int_0^{\frac{t}{\varepsilon^2} }f_0(y_0(\tau))\,\dif \tau,
\label{e.Weps}
\end{equation}
with rescaled fast dynamics $\dot{y}_0 = g_0(y_0)$. 
For $\int f_0(y)\mu(\dif y) = 0$ the integral is collecting weakly dependent variables with mean zero, provided the fast dynamics is sufficiently chaotic. The integral term in \eqref{e.Weps} is formally of the form of the CLT where $n=1/\eps^2$ terms are integrated and then scaled by $1/\sqrt{n}=\eps$, and, \highlight{assuming the CLT holds for $y$,} converges weakly to Brownian motion $W_t$ \highlight{on the long diffusive time scale $t={\mathcal{O}}(1)$}. The slow dynamics is approximated by
\begin{equation*}
\dif X = \sigma \dif W_t.
\end{equation*}
By explicitly calculating $\lim_{\eps\to 0} \int_\Lambda \mu(\dif y) W_\eps(t)^2 $, we obtain a Green-Kubo formula for the variance with
\begin{equation*}
\frac{1}{2}\sigma^2 = \int_0^\infty  \, \dif s \int_\Lambda \,  \mu(\dif y) \, f_0(y)f_0(\varphi^s y) ,
\end{equation*}
where $\varphi^s$ denotes the flow map of the fast dynamics $\dot{y}_0 = g_0(y_0)$.
Note that in the deterministic case, randomness is only introduced through the choice of the initial condition $y(0)$.\\

This heuristic argument can be made rigorous. It can be shown that a functional central limit theorem exists and solutions of the multi-scale system (\ref{eq.skewx})-(\ref{eq.skewy}) converge weakly to solutions the homogenized It\^o stochastic differential equation
\begin{equation}
d X = F(X) \dif t + \sigma(X) \, \dif W_t \, ,
\label{e.homog}
\end{equation}
where $W_t$ denotes $l$-dimensional Brownian motion.
The drift coefficient $F:\R^d\to\R^d$ is given by
\begin{align}
F(x) = &\int f_1(x,y)\,\mu(\dif y) + \int_0^\infty \dif s \int \E \left[ f_0(x,y)\cdot \partial_x f_0(x,\varphi^t y) \right] \,\mu(\dif y) \nonumber \\ &+ \int_0^\infty \dif s \int  \E\left[ g_1(x,y) \cdot \partial_y \left(f_0(x,\varphi^t y)\right) \right] \,\mu(\dif y) \label{eq.contF}  ,
\end{align}
with $\mu$ the ergodic invariant measure corresponding to $\dif y = g_0(y) \dif t + \beta(y) \dif W_t$ and $\E$ the expectation value w.r.t. the Wiener measure on $W_t$ in the stochastically driven case ($\beta \neq 0$). 
The diffusion coefficient $\sigma:\R^d\to\R^{d\times l}$ is given by the Green-Kubo formula
\begin{equation}
\hspace{-40pt}\sigma(x) \sigma^T(x) =\int_0^\infty \dif s\int \E\left( f_0(x,y) \otimes f_0(x,\varphi^ty) + f_0(x,\varphi^ty) \otimes f_0(x,y) \right)\, \mu(\dif y)  , \label{eq.contGK}
\end{equation}
where the outer product between two vectors is defined as $(a\otimes b)_{ij} = a_{i}b_{j}$. These expressions for the drift and diffusion can be derived formally by an asymptotic expansion of the backward Kolmogorov equation \footnote{Strictly speaking, these formulae are only valid for correlation functions which are slightly more than integrable. For fast systems with decaying autocorrelation functions which are only integrable, one can find expressions for the drift and diffusion coefficients which are, however, more complicated; see \cite{KellyMelbourne17} for details.}. \highlight{We remark that one can add a stochastic driver to the slow dynamics in (\ref{eq.skewx}) which would lead to an additively increased diffusion (\ref{eq.contGK}) (cf. \cite{PavliotisStuart}). For simplicity of exposition we do not consider this case here.}\\ 

In the deterministic case $\beta=0$, homogenization results \cite{MelbourneStuart11,GottwaldMelbourne13c,KellyMelbourne17} assure that the family of solutions of the original multi-scale dynamical system converge weakly in the sup-norm topology to the unique solution $X$ of the reduced stochastic differential equation (\ref{e.homog}) as $\varepsilon \to 0$. At the heart of diffusive limits of deterministic dynamical multi-scale systems lies a functional CLT which assures that $W_\eps(t) \to_w W(t)$ in $C([0,\infty),\R^d)$ in the limit $\eps \to 0$. The functional CLT implies the CLT but not vice versa. We are, however, only aware of a single example where a judiciously chosen unbounded observable of a deterministic dynamical system satisfies a CLT but not the associated functional CLT \cite{Gouezel04}.

In the following we evaluate corrections to the CLT for increments of $x$ by probing the finite $\eps$ corrections of the transition probability in an Edgeworth expansion.


\section{The Edgeworth expansion and corrections to the central limit theorem}
\label{sec.edge}
Before introducing Edgeworth expansions for continuous time multi-scale systems, it is instructive to briefly review the case of Edgeworth expansions in the discrete time stochastic case.\\ The central limit theorem describes when appropriately scaled sums $$S_n = \frac{1}{\sqrt{n}} \sum_{i = 1}^n (y_i - m)$$ of variables $y_i$ with mean $m$ and variance $\sigma^2$ converge in distribution to a normal distribution in the sense that
\begin{align*}
\PP(a\le S_n\le b) \to \frac{1}{\sqrt{2\pi\sigma^2}}\int_a^be^{-\frac{s^2}{2\sigma^2}}\,\dif s
\end{align*}
as $n\to\infty$ \cite{fellerbook,cinlarbook}. It is valid for {\em{i.i.d.}}\ random variables, weakly dependent random variables \cite{fellerbook,cinlarbook}, as well as for a large class of dynamical systems \cite{Dolgopyat04,MelbourneNicol05,MelbourneNicol08,MelbourneNicol09,Gouezel10}. Edgeworth expansions describe deviations from the CLT for finite $n$. We briefly review in the next subsection the well studied case of Edgeworth expansions for stochastic random variables.


\subsection{Edgeworth expansions for stochastic random variables}
In the stochastic context Edgeworth expansions are usually derived in the spectral framework by studying the characteristic function $\chi_n(\xi) =\mathbb{E} [\exp (i \xi S_n)]$ \ of the random variable $S_n$ ($\E$ denotes in this subsection the expectation with respect to $\mu$, the distribution of $y$). The characteristic function $\chi_n$ is related to the cumulants $c^{(j)}_n$ of $S_n$ by $\chi_n(\xi) = \exp(\sum_{j=1}^\infty \frac{(i\xi)^j}{j!}c^{(j)}_n )$. Assuming that $c^{(1)}_n = \mathbb{E}[S_n] = 0$, we have for weakly dependent random variables with uniform or strong mixing properties that $c^{(2)}_n = c^{(2)}_\infty + \frac{1}{n} \delta c^{(2)} + {o}(n^{-1})$, $c^{(3)}_n = \frac{1}{\sqrt{n}} c^{(3)}_\infty + {o}(n^{-1})$ and $c^{(4)}_n = \frac{1}{n} c^{(4)}_\infty + {o}(n^{-1})$, while higher cumulants are of higher order in $1/\sqrt{n}$. Specifically, up to order $o(n^{-1/2})$ the expansions of the cumulants are related to the correlation functions of the parent random variables $y_i$ by
\begin{align*}
\quad\;\,c_n^{(2)} &= \sigma^2 - 2\frac{1}{n} \sum_{j=1}^\infty j \E[y_1 y_{j+1}]\\
\sqrt{n} c_n^{(3)} &= \E[y_1^3] +3\sum_{j=1}^\infty \left(\E[y_1^2y_{j+1}] + \E[y_1y_{j+1}^2] \right) +6 \sum_{i,j=1}^\infty \E[y_1y_{1+i}y_{1+i+j}]
\end{align*}
with $\sigma^2=\E[y^2_1]+2\sum_{j=1}^\infty\E[y_1 y_j]$ given by the Green-Kubo formula \cite{GoetzeHipp83}. The above equations implicitly define the terms $c_\infty^{(2)}$, $\delta c^{(2)}$ and $c_\infty^{(3)}$. Inserting the expansions of the cumulants into the characteristic function, we obtain
\begin{align*}
\chi_n (\xi) & =  \exp \left( c^{(2)}_{n} \frac{(i \xi)^2}{2!} +
  c^{(3)}_{n} \frac{(i \xi)^3}{3!} + c^{(4)}_{n} \frac{(i \xi)^4}{4!} + \ldots \right)\\
  & = \exp \left(
  \frac{1}{\sqrt{n}} c^{(3)}_{\infty} \frac{(i \xi)^3}{3!} + \frac{1}{n}
  \left( \delta c^{(2)}  \frac{(i \xi)^2}{2!} + c^{(4)}_{\infty} \frac{(i
  \xi)^4}{4!} \right) + \mathcal{O} \left( \frac{1}{n^{3 / 2}} \right) \right)
  \times \exp \left( - c^{(2)}_{\infty} \frac{\xi^2}{2} \right) \\
  & = \left( 1 +
  \frac{1}{\sqrt{n}} c^{(3)}_{\infty} \frac{(i \xi)^3}{3!} + \frac{1}{n}
  \left( \delta c^{(2)}  \frac{(i \xi)^2}{2!} + c^{(4)}_{\infty} \frac{(i
  \xi)^4}{4!} \right) \right. \\
  & \left. \hspace{.2\textwidth} + \frac{1}{2} \left(\frac{1}{\sqrt{n}} c^{(3)}_{\infty} \frac{(i \xi)^3}{3!}\right)^2 + \mathcal{O} \left( \frac{1}{n^{3 / 2}} \right) \right) \times \exp \left( - c^{(2)}_{\infty} \frac{\xi^2}{2} \right)  .
\end{align*}
The characteristic function therefore converges pointwise to $\exp ( - c^{(2)}_{\infty} \xi^2/2 )$, which is the characteristic function of a Gaussian with variance $ c^{(2)}_\infty=\sigma^2$. Since, by L\'evy's continuity theorem, pointwise convergence of the characteristic functions is equivalent to convergence in distribution, the CLT follows. An expansion in $\frac{1}{\sqrt{n}}$ of the probability density function (pdf) of $S_n$ is obtained by the inverse Fourier transform of the characteristic function $\chi_n(\xi)$. Under the inverse Fourier transform the terms $(i \xi)^k$ in the expansion become $k$-th derivatives of the normal distribution, i.e. Hermite polynomials. Let $\rho_n$ be the pdf of the normalized Birkhoff sum $S_n$. Then, the first two Edgeworth approximations $\rho^{(1)}_{n}$ and $\rho^{(2)}_{n}$ of the probability density $\rho_n$ of $S_n$ are
\begin{align}
  \rho^{(1)}_{n} (\xxi) & =  \mathbf{n}_{0,\sigma^2} (\xxi)\left(1 + \frac{c^{(3)}_\infty}{6 \sigma^3
  \sqrt{n}} H_3 \left(\frac{\xxi}{\sigma}\right)\right) , 
\label{eq.edgeworthDens}
\end{align}
and
\begin{align}
  \rho^{(2)}_{n} (\xxi) & =  \mathbf{n}_{0,\sigma^2} (\xxi)\left(1 
  + \frac{c^{(3)}_\infty}{6 \sigma^3 \sqrt{n}} H_3 \left(\frac{\xxi}{\sigma}\right)
   + \frac{\delta c^{(2)}}{2 \sigma^2 n} H_2 \left(\frac{\xxi}{\sigma}\right)  + \frac{c^{(4)}_\infty}{24 \sigma^4 n} H_4 \left(\frac{\xxi}{\sigma}\right)
   + \frac{{c^{(3)}_\infty}^2}{72 \sigma^6 n} H_6 \left(\frac{\xxi}{\sigma}\right)
   \right) , \nonumber
\end{align}
where $H_k$ is the $k$-th Hermite polynomial. We have $\rho_n(\xxi) = \rho^{(1)}_{n} (\xxi) + o \left( \frac{1}{\sqrt{n}} \right)$ and $\rho_n(\xxi) = \rho^{(2)}_{n} (\xxi) + o \left( \frac{1}{{n}} \right)$ uniformly in $\xxi$ {\cite{fellerbook}}. The Edgeworth approximation generally yields an improved approximation of the pdf around the mean of the distribution \cite{FieldBook}. Note that $\rho^{(p)}_n$ is no longer nonnegative or normalized and therefore the approximation $\rho^{(p)}_{n}$ is no longer a probability density function. In contrast to the Gaussian distribution $\mathbf{n}_{0, \sigma^2} (\xxi)$, the first Edgeworth approximation may have a non-zero third moment $c^{(3)}_\infty / \sqrt{n}$, which vanishes as $n \rightarrow \infty$. Higher order approximations $\rho^{(p)}_{n}$ can be derived by increasing the order of the Taylor series. 

Interestingly, the functional form of the expansion is universal, in the sense that the parent process $y_i$ only enters through the asymptotic cumulants $c^{(2)}_\infty=\sigma^2$, $\delta c^{(2)}$, $c^{(3)}_\infty$ and $c^{(4)}_\infty$. 

In the stochastic context, Edgeworth expansions have been obtained for time series such as ARMA processes \cite{GoetzeHipp94}, continuous-time diffusions \cite{ait-sahalia_maximum_2002} and, employing the Nagaev-Guivarc'h method for the characteristic function, for ergodic Markov chains \cite{HervePene10}. Edgeworth expansions have, to the best of our knowledge, not been explored in the multi-scale context. The next sections aim at filling this gap.


\section{Edgeworth expansions for continuous-time multi-scale systems}
\label{sec.Econt}
We now determine Edgeworth corrections for multi-scale systems. The formulae will be given in terms of the generator of the fast process
\begin{align}
\mathcal{L}_0 = g_0(y) \cdot \partial_y + \frac{1}{2} \beta(y) \beta^T(y) : \partial_y \cdot \partial_y
\end{align}
with $L^2$-adjoint $\mathcal{L}^*_0$ which acts as 
\begin{align}
\mathcal{L}^*_0 \nu = - \partial_y \cdot (g_0(y) \nu) + \frac{1}{2} \partial_y\cdot \partial_y\cdot  \left( \beta(y)\beta^T(y)\, \nu\right) .
\label{e.L0star}
\end{align}
As described in the introduction, multi-scale systems \ds\ are characterized by three distinct time scales: the fast time scale of $\Ord(\eps^2)$, an intermediate time scale of $\Ord(\eps)$ on which the slow dynamics is trivial but the fast dynamics has equilibrated, and the long diffusive time scale of $\Ord(1)$ on which the slow dynamics exhibits nontrivial diffusive behaviour.
The particular r\^ole of the intermediate time scale to control the normality of the noise is formally reflected in the homogenized stochastic limit system (\ref{e.homog}) which evolves on the diffusive time scale of $\Ord(1)$; since $\dif W_t$ scales like $\sqrt{t}$ the Brownian motion is dominant on time scales of $\Ord(\eps)$. 
On this time scale, the transition probability $\pi_h$ of the diffusive homogenized limit system (\ref{e.homog}) satisfies the CLT and converges according to
\begin{align*}
\pi_h(\xxi,t,x_0) = \PP\left( \left. \frac{X(t)-x_0}{\sqrt{t}}  \in \left( \xxi, \xxi + \mathrm{d}\xxi \right) \right| X(0)=x_0 \right)\to\mathbf{n}_{0,\sigma^2}(\xxi) ,
\end{align*}
where we let ${t \to 0}$ since $t$ is of $\Ord(\eps)$ and where $X$ solves the homogenized stochastic differential equation (\ref{e.homog}). To study deviations from the Gaussian behaviour of the limit $\eps\to 0$, we develop in Appendix A a semi-group formalism to calculate the Edgeworth expansion of the intermediate-time transition probabilities of the slow variable $x$ of the multi-scale system \ds\ 
\begin{equation*}
\pi_\eps(\xxi,t,x_0) = \PP \left( \left. \frac{x(t) - x_0 }{\sqrt{t}} \in (\xxi, \xxi + \mathrm{d} \xxi) \right| x(0) = x_0, \, y(0) \sim   \mu_{x_0}^{(\eps)}\right),
\end{equation*}
for $t=\eps$, and expand $\pi_{\eps} ( \xxi, \eps, x_0 ) $ in $\sqrt{\eps}$. Here $\mu_{x_0}^{(\eps)}(y)$ is the invariant measure $\mu^{(\eps)}(x,y)$ of \eqref{eq.skewx}-\eqref{eq.skewy} conditioned on $x=x_0$. By the Rokhlin disintegration theorem \cite{rokhlin1949fundamental} the conditional measure is essentially unique and we furthermore assume that $\mu^{(\eps)}_{x_0}$ obeys linear response w.r.t. $\varepsilon$, such that the conditional measure can be expanded in $\varepsilon$ around $\mu$. Whereas the coefficients in the homogenized equation only involve the two-time statistics, higher-order Edgeworth expansions involve higher-order time correlations, containing more information about the dynamics.\\


As in the case of random variables described in the previous section, the Edgeworth expansion of the intermediate-time transition probability $\pi_\eps$ can be calculated by the asymptotic expansion of its associated characteristic function which is entirely determined by the cumulants. We therefore set out to asymptotically calculate the $p^{\rm{th}}$-moments of the slow variables $(x(t) - x_0)/\sqrt{t}$ in orders of $\sqrt{\eps}$ 
\begin{align}
m^{(p)} = \E^{x_0,\mu} \left[ \left( \frac{\hat{x}(t)}{\sqrt{t}} \right)^p \right]
,
\label{e.exp0}
\end{align}
with $\hat x(t) = x(t) - x_0$ on the intermediate time scale $t=\eps$ and the corresponding cumulants $c^{(p)}$. The conditional average $\E^{x_0, \mu}$ 
is with respect to the product measure
\begin{equation}
{\mu}^{(\eps)}(d x, d y) = \left( \delta_{x_0} \times {\mu}^{(\eps)}_{x_0} \right) (d x, d y).
\label{e.mueps}
\end{equation}
The measure $\mu^{(\eps)}_{x_0}$ depends on $\eps$ for non-skew product systems where the fast dynamics depends on the slow dynamics. Assuming that $\mu^{(\eps)}_{x_0}$ obeys linear response w.r.t. $\eps$ we expand the measure ${\mu}^{(\varepsilon)}_{x_0}$ as
\begin{equation}
{\mu}^{(\eps)}_{x_0} ={\mu}^{(0)}_{x_0} + \eps {\mu}^{(1)}_{x_0} + {\mathcal{O}}(\eps^2) , \label{eq.linres}
\end{equation}
with ${\mu}^{(0)}_{x_0} = \mu$ the invariant measure of the fast process $\dif {y} = g_0(y) \dif t + \beta(y) \dif W_t$. We note that, while linear response has been proven for a wide class of deterministic and stochastic systems and has been used for model reduction \cite{wouters_multi-level_2013,wouters_disentangling_2012,wouters_parameterization_2016}, counter-examples do exist \cite{BaladiSmania08,BaladiSmania10,Baladi14,BaladiEtAl15,DeLimaSmania15,GottwaldEtAl16}.
\\
We seek expansions of the scaled moments $m^{(p)}(t)$ in $\eps$ and in $t$. We take the limit $t=\varepsilon \ll 1$ and $t/\varepsilon^2 \to \infty$ as $\varepsilon \to 0$. In this limit the random variable $\hat{x}(t)/\sqrt{t}$ converges to a normal distribution and the coefficients of the different powers of $\sqrt{\eps}$ appearing in the expansions of its cumulants will provide the desired Edgeworth corrections coefficients. We expand the $p^{\rm{th}}$ moments as
\begin{equation}
m^{(p)} = m_0^{(p)}+\sum_{|{\bmalpha}|>0} \eps^{\alpha_\eps} t^{\alpha_t}m_{\bm\alpha}^{(p)} ,
\label{e.mexp}
\end{equation}
where we use multi-index notation to denote the expansions in $\eps$ and $t$ with ${\bmalpha}=(\alpha_\eps,\alpha_t)$ and with $|{\bmalpha}|=\alpha_\eps+\alpha_t$ being the combined order of the contribution. We similarly expand the cumulant as
\begin{equation}
c^{(p)} = c_0^{(p)}+\sum_{|{\bmalpha}|>0} \eps^{\alpha_\eps} t^{\alpha_t}c_{\bm\alpha}^{(p)}.
\label{e.cexp}
\end{equation}
Note that $\alpha_{\eps,t}<0$ and $\alpha_{\eps,t}$ half-integer is allowed.\\ 
Our expansion of the transition probability $\pi_\eps(\xxi,t=\eps,x_0)$ (\ref{e.gamma_xi}) involves the first four cumulants. The derivation of those cumulant expansions can be found in Appendix \ref{appendix}. Here we only state the resulting formulae. The first cumulant is given up to order $\mathcal{O}(\eps^{\frac{3}{2}})$ by
\begin{align}
c^{(1)} = \sqrt{t} \, c^{(1)}_{0,\frac{1}{2}} (x_0) + R^{(1)}_\eps ,
\label{e.m1}
\end{align}
where the remainders $R^{(j)}_\eps = \sum_{| {\bmalpha}|>1} \eps^{\alpha_\eps} t^{\alpha_t} c^{(j)}_{\bmalpha}$ for $j=1,\cdots,p$ consist of higher order terms  and  
\begin{align}
c^{(1)}_{0,\frac{1}{2}} = F(x)=\langle f_1 \rangle  - \langle f_0 \mathcal{L}_{0\perp}^{-1} \partial_x f_0 \rangle - \langle (g_1 \partial_y) \mathcal{L}_{0\perp}^{-1} f_0 \rangle ,
\label{e.m1b}
\end{align}
recovering the drift coefficient $F(x)$ of the homogenized equation \eqref{e.homog} (cf. \cite{GivonEtAl04,PavliotisStuart}). The angular brackets denote the conditional average with respect to $\mu^{(0)}_{x_0} = \mu$, i.e. $\langle A(x,y) \rangle = \int A(x,y) \mu (\dif y)$. The second cumulant and moment is given up to order $\mathcal{O}(\eps^{\frac{3}{2}})$ by
\begin{align}
c^{(2)} = m^{(2)} &= c^{(2)}_0 +  t\,  c_{0,1}^{(2)} + \frac{\eps^2}{t}c^{(2)}_{2,-1} + \eps\,  c^{(2)}_{1,0} + R^{(2)}_\eps.
\label{e.m2}
\end{align}
The $\Ord(1)$ contribution is given by the homogenized Green-Kubo formula \eqref{eq.contGK}
\[
 c^{(2)}_0 = \sigma^2 = - 2 \la f_0 \mathcal{L}_{0\perp}^{-1} f_0\ra 
\]
and higher-order contributions are given by
\begin{align}
c^{(2)}_{0,1}
&= \frac{1}{2} \, \sigma^{2} \left( \frac{\partial \sigma}{\partial x} \right)^{2} + \frac{1}{2} \,  \sigma^{3} \frac{\partial^{2}\sigma}{\partial x^{2}} +  \sigma^{2} \frac{\partial F}{\partial x} +  F \sigma \frac{\partial \sigma}{\partial x} +  F^{2} \label{e.m2c}
\\
c^{(2)}_{2,-1} &= 
- 2  \langle f_0 \mathcal{L}_{0\perp}^{-2} f_0 \rangle
\label{e.m2a}
\\
c^{(2)}_{1,0} &=
 - 2  \langle f_0 \mathcal{L}_{0\perp}^{-1} f_1 \rangle  
 - 2 \langle f_1 \mathcal{L}_{0\perp}^{-1} f_0 \rangle 
 + 2  \langle f_0 \mathcal{L}_{0\perp}^{-1} \partial_x f_0 \mathcal{L}_{0\perp}^{-1} f_0 \rangle 
 \nonumber\\ 
&\quad \,  + 4 \langle f_0 \mathcal{L}_{0\perp}^{-1} f_0 \mathcal{L}_{0\perp}^{-1} \partial_x f_0 \rangle 
+ 2 \langle f_0 \mathcal{L}_{0\perp}^{-1} (g_1 \partial_y) \mathcal{L}_{0\perp}^{-1} f_0 \rangle 
\nonumber\\
&\quad \, + 2 \langle (g_1 \partial_y) \mathcal{L}_{0\perp}^{-1} f_0 \mathcal{L}_{0\perp}^{-1} f_0 \rangle .
\label{e.m2b}
\end{align}
Here $\mathcal{L}_{0\perp}^{-1}$ denotes the invertible operator whose inverse is the restriction of $\mathcal{L}_0$ to the space orthogonal to the projection onto the invariant measure $\mu_{x_0}^{(0)}$ (see Section (\ref{sec.laplace}) in the Appendix for more details; cf. (\ref{e.L0inv})). Recall that $\mathcal{L}_{0}$, associated with the ergodic fast dynamics, has a nontrivial kernel, namely functions which are constant in $y$, and as such is noninvertible.

The third moment and its cumulant are given up to order $\mathcal{O}(\eps^{\frac{3}{2}})$ by
\begin{align}
c^{(3)} = m^{(3)} = \sqrt{t}\, c^{(3)}_{0,{\tiny{\frac{1}{2}}}} + \frac{\eps}{\sqrt{t}}\, c^{(3)}_{1,{-\tiny{\frac{1}{2}}}} + R^{(3)}_\eps
\label{e.m3}
\end{align}
with
\begin{align}
c^{(3)}_{0,{\tiny{\frac{1}{2}}}} &=
6\langle f_0 \mathcal{L}_{0\perp}^{-1} f_0 \rangle \frac{\partial}{\partial x}\langle f_0 \mathcal{L}_{0\perp}^{-1} f_0 \rangle 
\label{e.c3a}
\\
c^{(3)}_{1,{-\tiny{\frac{1}{2}}}} &=
6  \left\langle f_0 \mathcal{L}_{0\perp}^{-1} f_0 \mathcal{L}_{0\perp}^{-1} f_0 \right\rangle
\label{e.c3b}
\end{align}
and the  fourth cumulant is given up to order $\mathcal{O}(\eps^{\frac{3}{2}})$ by
\begin{align}
c^{(4)} =  t\,c^{(4)}_{0,1} +   \eps\,c^{(4)}_{1,0} +\frac{{{\varepsilon}}^{2}}{t}\,  c^{(4)}_{2,-1} +  R^{(4)}_\eps ,
\label{e.m4}
\end{align}
with
\begin{align}
c^{(4)}_{0,1} &= 
-24 \, \langle f_0 \mathcal{L}_{0\perp}^{-1} f_0 \rangle \left( \frac{\partial}{\partial x}\langle f_0 \mathcal{L}_{0\perp}^{-1} f_0 \rangle \right)^{2} 
- 16 \, \langle f_0 \mathcal{L}_{0\perp}^{-1} f_0 \rangle^{2} \frac{\partial^{2}}{\partial x^{2}}\langle f_0 \mathcal{L}_{0\perp}^{-1} f_0 \rangle 
\label{c4a} \\
c^{(4)}_{1,0} &=
- 24 \,  \langle \frac{\partial}{\partial x} f_0 \mathcal{L}_{0\perp}^{-1} f_0 \mathcal{L}_{0\perp}^{-1} f_0 \rangle \langle f_0 \mathcal{L}_{0\perp}^{-1} f_0 \rangle 
- 36 \, \langle f_0 \mathcal{L}_{0\perp}^{-1} f_0 \mathcal{L}_{0\perp}^{-1} f_0 \rangle \frac{\partial}{\partial x}\langle f_0 \mathcal{L}_{0\perp}^{-1} f_0 \rangle \label{c4b}
\qquad\\
c^{(4)}_{2,-1}  &=
24  \left( \langle f_0 \mathcal{L}_{0\perp}^{-2} f_0 \rangle \langle f_0 \mathcal{L}_{0\perp}^{-1} f_0 \rangle - \langle f_0 \mathcal{L}_{0\perp}^{-1} f_0 \mathcal{L}_{0\perp}^{-1} f_0 \mathcal{L}_{0\perp}^{-1} f_0 \rangle \right) \label{c4c}
.
\end{align}

Higher cumulants give rise to terms of at least $\mathcal{O}(\eps^{\frac{3}{2}})$. 
As expected from the CLT, the only $\mathcal{O}(1)$ contribution to the cumulants (\ref{e.m1})-(\ref{e.m4}) appears in the second cumulant (\ref{e.m2}). The higher-order terms determine the Edgeworth corrections. At $\Ord(\sqrt{\eps})$ -- describing the lowest order correction to the CLT -- the first and third cumulant (\ref{e.m3}) feature, and, if non-zero, give rise to skewness. 
Corrections to the variance and the fourth-order cumulant $c^{(4)}$ start to contribute at order ${\mathcal{O}}(\eps)$. Note that the feedback of the slow dynamics on the approach to Gaussianity via vector field $f_1$ only appears in the corrections to the first moment. 

In this formulation of the cumulant expansion we have not yet substituted $t=\eps$. This allows us to separate the contributions which pertain in the limit $\eps \to 0$ and encode the Edgeworth corrections to the transition probabilities of the homogenized slow diffusion equation \eqref{e.homog}, namely those with $\alpha_\eps = 0$ (see also \cite{ait-sahalia_maximum_2002}).
\begin{cor}
For the transition probabilities $\pi_h(\xxi,t,x_0)$ of the homogenized system, we have the following expansion for small $t$
\begin{align}
 &\pi_h(\xxi,t,x_0) = \mathbf{n}_{0,\sigma^2} \left({\xxi}\right) \left( 1 
        + \sqrt{t} \left( \frac{c^{(1)}_{0,\frac{1}{2}}}{\sigma} H_1\left(\frac{\xxi}{\sigma}\right) + \frac{c^{(3)}_{0,\frac{1}{2}}}{3!\sigma^3} H_3\left(\frac{\xxi}{\sigma}\right) \right) \right. \label{e.gamma_xi_h} \\
         &\quad \left. + t \left(\frac{ c^{(2)}_{0,1} + {c^{(1)}_{0,\frac{1}{2}}}^2}{2 \sigma^2} H_2\left(\frac{\xxi}{\sigma}\right) + \frac{c^{(4)}_{0,1} + 4 c^{(1)}_{0,\frac{1}{2}} c^{(3)}_{0,\frac{1}{2}}}{4! \sigma^4} H_4\left(\frac{\xxi}{\sigma}\right) + \frac{{c^{(3)}_{0,\frac{1}{2}}}^2}{2 (3! \sigma^3)^2} H_6\left(\frac{\xxi}{\sigma}\right) \right)
        \right) \nonumber \\
        &\hspace{10cm} + \Ord(t^{\frac{3}{2}}),
        \nonumber       
\end{align}
where $H_n(\xxi) = (\xxi - \frac{\dif}{\dif \xxi})^n 1$ are Hermite polynomials of degree $n$.
\end{cor}
The remaining contributions, \eqref{e.m2a},\eqref{e.m2b},\eqref{e.c3b},\eqref{c4b} and \eqref{c4c} involve intricate correlations between the slow and the fast dynamics which do not vanish on the intermediate time scale for $\eps \to 0$ and $t\neq 0$. In particular the homogenized limit $\eps \to 0$ does not involve the slow vector field $f_1$. 

In homogenization the knowledge of the drift $F$ and the Green-Kubo diffusion $\sGK$ is sufficient to determine the effective reduced stochastic slow dynamics (cf. (\ref{e.homog})). Going beyond the CLT requires higher-order Edgeworth corrections involving indefinite integrals over multi-point correlation functions. It is instructive to note that the vector field $f_1(x,y)$ nontrivially enters the corrections to the variance (\ref{e.m2}). The detailed proof of our main result (\ref{e.gamma_xi}) can be found in Appendix \ref{appendix}.

\section{Numerical results}
\label{sec.numerics}
We now corroborate our main result on the explicit formula (\ref{e.gamma_xi}) for the transition probability $\pi_\eps(\xxi,t=\eps,x_0)$ with numerical simulations. We show that the Edgeworth approximation of the intermediate-time transition probability provides a much better approximation to the true transition probability of the multi-scale system with finite time scale separation than the Gaussian limiting distribution implied by the assumption of infinite-time-scale separation and homogenization theory.\\ 
We consider here as an example of the general multi-scale system \ds\/, the following skew-product slow-fast system for a one-dimensional slow variable
\begin{align}
  \dot{x} = \frac{1}{\eps} f_0(y)  - V'(x) \label{eq.sfL63x} 
\end{align}
driven by a fast Lorenz system   
\begin{align}
  \dot{y_1} &= \frac{s_l}{\eps^2} (y_2-y_1) \label{eq.sfL63y1} \\
  \dot{y_2} &= \frac{1}{\eps^2} (y_1 (r_l - y_3) - y_2) \label{eq.sfL63y2}\\
  \dot{y_3} &= \frac{1}{\eps^2} (y_1 y_2 - b_l y_3) \label{eq.sfL63y3}  ,
\end{align}
where we choose the classical parameters $s_l=10$, $r_l=28$ and $b_l=8/3$ \cite{Lorenz63}. The Lorenz system is rapidly mixing with exponentially decaying correlation \cite{ArujoEtAl15}. The fast variables drive the slow variable $x$ via a function $f_0(y) = \cos(y_2/2) \exp(y_1/40) -\bar f_0$, generating skewed deterministic noise. Here $\bar f_0$ is a constant chosen such that the expectation value of $f_0$ under the fast dynamics is zero and the system \eqref{eq.sfL63x}--\eqref{eq.sfL63y3} satisfies the centering condition. We consider here a harmonic potential $V(x)= \alpha x^2$ \textcolor{blue}{(i.e. $f_1=V'(x)$)}. In the absence of coupling to the fast Lorenz system the slow dynamics settles to the stable fixed point $x=0$. With nontrivial coupling $f_0$, the fast dynamics induces fluctuations around the stable fixed point. An example of a slow trajectory is given in Figure~\ref{fig.slowx}. On large enough time scales we observe seemingly stochastic behaviour (left-hand figure), similar to that of the limiting homogenized system. On shorter time scales, the smoothness of the noise is however still discernible (right-hand figure).

\begin{figure}[h]
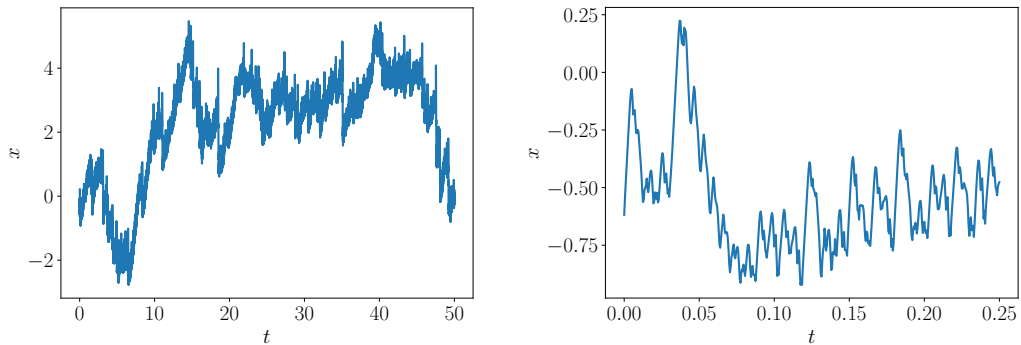

	\centering
	\begin{tabular}{c c}
		\includegraphics[width=.475\textwidth]{{{sfL63_quadratic_slow_trajectory}}} & 
		\includegraphics[width=.5\textwidth]{{{sfL63_quadratic_slow_trajectory_short}}}
	\end{tabular}
	
	\caption{Trajectory of the slow variable $x$ of the slow-fast system \eqref{eq.sfL63x}-\eqref{eq.sfL63y3} with $\alpha=0.1$ and $\eps=0.1$.\label{fig.slowx}}
\end{figure}

To numerically estimate the transition probability $\pi_\eps(\xxi,t=\eps,x_0)$ we perform ensemble averages of the slow variable $x(t=\eps)$ using $M=10^5$ long time simulations of $6\times 10^7$ time units using a Dormand-Prince order $4/5$ Runge-Kutta method \cite{odejl,dormand_family_1980}. The simulations have the same initial condition $x(0)=1$ of the slow variable but differ in the initial conditions for the fast Lorenz system which are chosen randomly from the Lorenz attractor (this is assured by letting random initial conditions settle on the attractor after a transient period of $10$ time units).

The limiting Gaussian transition probability $\pi_h$ can be analytically determined from the homogenized limiting equation $\mathrm{d} X = - \alpha X \mathrm{d}t + \sigma \mathrm{d}W$ which is an Ornstein-Uhlenbeck process where the diffusion coefficient $\sigma^2$ is given by the Green-Kubo formula. The limiting Gaussian transition probability has mean $x_0 \exp(-\alpha t)$ and variance by $(1 - \exp(-2 \alpha t)) \sGK/(2\alpha) $. The Edgeworth corrections to this limiting distribution are given by (\ref{e.gamma_xi}) and the first four cumulants  (\ref{e.m1}), (\ref{e.m2}), (\ref{e.m3}) and (\ref{e.m4}). The terms $c^{(2)}_{2,1}$ \eqref{e.m2a}, $c^{(3)}_{1,\frac{1}{2}}$ \eqref{e.c3b} and $c^{(4)}_{2,-1}$ \eqref{c4c} are estimated numerically by setting the gradient force $V^\prime(x)\equiv 0$ in the long time integrations. Setting $f_1 = V^\prime(x)\equiv 0$ implies that the Edgeworth expansion only involves  \eqref{e.m2a}, \eqref{e.c3b} and \eqref{c4c} which in turn can be accurately estimated without being numerically dominated by the other terms involving $f_1$. For details on the calculations of these coefficients see \cite{wouters_stochastic_2018}.

Figure~\ref{fig.ptrans} shows a comparison of the transition probabilities of the full deterministic multi-scale system with $\eps=0.1$, the limiting Gaussian implied by the homogenized limit and the transition probability given by the Edgeworth expansion. It is clearly seen that the the homogenized limit system is not able to capture the inherent skewness of the slow dynamics, whereas the Edgeworth approximation capture the transition probability remarkably well. We also see that the Edgeworth approximation is only a good approximation of the transition probability for the slow variable near the mode and exhibits deviations for $\xxi$-values far from the mean with unphysical negative values. 
\begin{figure}[h]
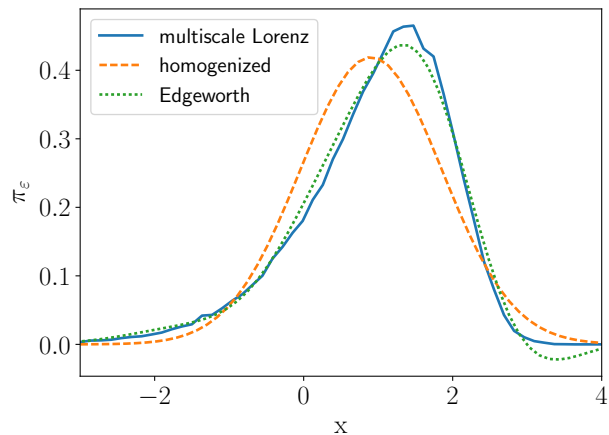

\centering
\begin{tabular}{c c}
  \includegraphics[height=6cm]{{{ptrans}}}
\end{tabular}

\caption{Transition probability $\pi_\eps(\xxi,t,x_0)$ of the slow variable for the full multi-scale Lorenz system \eqref{eq.sfL63x}-\eqref{eq.sfL63y3} with $\alpha=0.1$ and $\eps=0.1$ (blue solid), transition probability $\pi_h(\xxi,t,x_0)$ of the corresponding homogenized system (orange dashed) and $\pi_\eps(\xxi,t,x_0)$ of the Edgeworth expansion \eqref{e.gamma_xi} (green dotted). The transition probabilities are  evaluated at $t=1$ and were initiated at $x_0=1$.
\label{fig.ptrans}}
\end{figure}


\section{Summary and outlook}
\label{sec.summary}

In this article we derived Edgeworth expansions that describe corrections to the Gaussian limiting behaviour of slow-fast systems for finite time scale separation. The Edgeworth expansion is achieved using a semi-group formalism for the transfer operator, where a Duhamel-Dyson series is used to asymptotically determine the corrections at any desired order of the time scale parameter $\eps$. The corrections appear on the intermediate time scale $\mathcal{O}(\eps)$ and the asymptotics requires the limit $t=\eps\ll 1$ and $t/\eps^2\to \infty$. We developed a diagrammatic representation of higher-order correlation integrals to control the combinatorial wealth of their asymptotic expansion in $\eps$. To obtain our explicit formula for the transition probability (\ref{e.gamma_xi}) we required mixing assumptions on the integrability of higher-order autocorrelation functions. It is pertinent to mention that homogenization theory does not rely on mixing and the diffusive limit equations exist for non-mixing systems (albeit not with a Green-Kubo formula (\ref{eq.contGK}) for the diffusion), and similarly there are mixing systems which do not allow for a diffusive limit. We expect that similarly one can derive an Edgeworth expansion without these strong mixing assumptions, at the cost of not having compact explicit expressions for the drift and diffusion.\\

Our work points to several applications and directions, planned for further research. 
We derive here the Edgeworth expansion for continuous-time multi-scale systems in the case where at leading order the slow dynamics does not couple back into the fast dynamics, i.e. $g_0=g_0(y)$. 
We also expect a similar expansion to hold when $g_0=g_0(x,y)$ and the slow dynamics couples back into the fast dynamics at leading order. A complicating issue here is the potential breakdown of linear response when the fast invariant measure does not depend smoothly on the slow variables. 

Edgeworth corrections to the CLT are not restricted to systems where the noise originates as the accumulative effect of rapidly decorrelating fast variables, as we have described here for slow-fast systems. Sums of uncorrelated deterministic variables also appear in weak coupling limits where a distinguished degree of freedom is weakly coupled to a bath of $N$ degrees of freedom. Here stochastic limit systems arise in the limit of an infinitely large bath with $N\to\infty$ \cite{GivonEtAl04}. Again, corrections for finite values of $N$ can be studied using an Edgeworth expansion. 

The universal form of the deviations from Gaussianity given by the Edgeworth expansion suggests that one can devise stochastic parametrizations and effective diffusive dynamics for the slow variables for finite time scale separation by substituting the fast dynamics with a surrogate system with the same Edgeworth coefficients, improving on the classical homogenization limit SDEs. In particular, the universal character allows for a data-driven approach where the cumulants are numerically estimated to build a stochastic model for the observed variables.


\enlargethispage{20pt}


\dataccess{This article has no additional data.}

\aucontribute{JW designed the research and performed the numerical experiments. All authors contributed to the research and the writing of the paper.}

\competing{There are no competing interests.}

\funding{The research leading to these results has received funding from the European Community's Seventh Framework Programme (FP7/2007-2013) under grant agreement n\textdegree\,  PIOF-GA-2013-626210. GAG is partially supported by the ARC grant DP180101385.}

\ack{We thank Ben Goldys and Fran\c{c}oise P\`{e}ne for enlightening discussions and comments.}




\appendix

\section{Appendix: Derivation of the cumulant expansion}
\label{appendix}
{\bf{Outline of the derivation:}} The Edgeworth expansion involves an expansion of the cumulants in orders of $\eps$. The derivation of the Edgeworth expansion proceeds in a number of steps. We first consider the Laplace transform of the moments, when expressed in terms of the Koopman operator, and perform a subsequent expansion using eigenfunctions of the fast generator  $\mathcal{L}_0$ (Section~\ref{sec.laplace}). Introducing a two-dimensional diagrammatic representation we perform a combinatorial accounting of the sequences appearing in the expansion of the Laplace transform (Section~\ref{sec.nonzero}). In particular, this allows us to identify those sequences which vanish and those which will give non-trivial contributions. This will allow us to provide the explicit expressions of the cumulants $c^{(p)}$ with $p\le 4$ stated in (\ref{e.m1}), (\ref{e.m2}), (\ref{e.m3}) and (\ref{e.m4}) to capture the Edgeworth corrections up to $\Ord(\eps^{\frac{3}{2}})$ (Section~\ref{sec.calcmoments}). We then show that it is sufficient for the estimation of the cumulants up to $\Ord(\eps^{\frac{3}{2}})$ to perform averages with respect to the fast measure $\mu^{(0)}_{x_0}$ only, discarding contributions from the linear response term $\mu^{(1)}_{x_0}$ (Lemma~\ref{Lemma_LinResp}). We conclude with Lemma~\ref{Lemma_HigherOrder} showing that higher cumulants $c^{(p)}$ with $p\ge 5$ do not contribute at order $\Ord(\eps^{\frac{3}{2}})$ to the Edgeworth expansion and the transition probability $\pi_\eps(\xxi,t=\eps,x_0)$. \\

{\bf{Derivation of the cumulant expansion:}} We express the conditional average in \eqref{e.exp0} in terms of the Koopman operator $e^{\mathcal{L}t}$ associated with the multi-scale system (\ref{eq.skewx})-(\ref{eq.skewy}) 
with the generator 
\begin{align}
\mathcal{L} =&\frac{1}{\varepsilon^2}\mathcal{L}_0 + \frac{1}{\varepsilon}\mathcal{L}_1 + \mathcal{L}_2, 
\end{align}
where
\begin{align}
\mathcal{L}_0  &= g_0(y) \cdot \partial_y + \frac{1}{2} \beta(y)\beta^T(y) : \partial_y \cdot \partial_y ,\quad
\mathcal{L}_1  &= f_0(x,y) \cdot \partial_x  + g_1(x,y) \cdot  \partial_y ,\quad 
\mathcal{L}_2  &= f_1(x,y) \cdot \partial_x  .
\end{align}
We seek expansions in $t$ and $\varepsilon$ of the moments
\begin{align}
  m^{(p)} 
  & =  \frac{1}{t^{p / 2}} \mathbb{E}^{x_0,{\mu}} \left[
  e^{t \left( \frac{\mathcal{L}_0}{\varepsilon^2} +
  \frac{\mathcal{L}_{12}}{\varepsilon} \right)} z^{(p)} (x) \right] 
 =  \frac{1}{t^{p / 2}} 
  \int \int e^{t \left( \frac{\mathcal{L}_0}{\varepsilon^2} + \frac{\mathcal{L}_{12}}{\varepsilon} \right)} z^{(p)} (x)  \left( \delta_{x_0} \times \mu_{x_0}^{(\eps)} \right)(\dif x, \dif y) ,
\label{eq.momentp}
\end{align}
where $z^{(p)} (x) = (x - x_0)^p$ and where we introduce for convenience $\mathcal{L}_{12} = \mathcal{L}_1 + \varepsilon \mathcal{L}_2$. Using (\ref{eq.linres}), we split the conditional average over the invariant measure $\mu^{(\eps)}$ (see (\ref{e.mueps})) according to
\begin{align}
m^{(p)} &= \frac{1}{t^{p/2}} \delta_{x_0} \left( (\mu_{x_0}^{(0)} + \eps \mu_{x_0}^{(1)} +\mathcal{O}(\eps^2) ) e^{t \mathcal{L}} z^{(p)}(x) \right) 
= \frac{1}{t^{p/2}} \delta_{x_0}\left( \mathcal{A}^{(p)} + \eps\, \mathcal{B}^{(p)}+\mathcal{O}(\eps^2) \right) ,
\label{eq.momentp2}
\end{align}
where we define
\begin{equation*}
\mathcal{A}^{(p)}(t)  = \left\langle e^{t \left( \frac{\mathcal{L}_0}{\varepsilon^2} +
\frac{\mathcal{L}_{12}}{\varepsilon} \right)} z^{(p)} (x)
\right\rangle 
\quad {\rm{
and
}}
\quad
\mathcal{B}^{(p)}(t)  = \mu^{(1)}_{x_0} \left( e^{t \left( \frac{\mathcal{L}_0}{\varepsilon^2} +
\frac{\mathcal{L}_{12}}{\varepsilon} \right)} z^{(p)} (x)
\right) .
\end{equation*}
As before, angular brackets denote the average with respect to $\mu_{x_0}^{(0)}$.

We first calculate the first term in (\ref{eq.momentp2}) associated with the average over $\mu_{x_0}^{(0)}$ before proving that averages with respect to $\mu^{(1)}_{x_0}$ do not contribute in Lemma~\ref{Lemma_LinResp}.

\subsubsection{Laplace transform of the moments and expansion using eigenfunctions of $\mathcal{L}_0$}
\label{sec.laplace}

We assume that the spectrum of the generator $\mathcal{L}$ is contained in the left-half complex plane, $\lbrace z \in \mathbb{C}: \mathrm{Re}(z) \leq 0 \rbrace$ (see \cite{butterley_smooth_2007} for Anosov flows). We can then take the Laplace transform of $\mathcal{A}^{(p)} (t)$
\begin{align}
\mathfrak{L} \{ \mathcal{A}^{(p)} \} (s) & =  \left\langle \left( s
\1 - \frac{\mathcal{L}_0}{\varepsilon^2} -
\frac{\mathcal{L}_{12}}{\varepsilon} \right)^{- 1} z^{(p)} (x)
\right\rangle   \label{eq.laplace0}
\end{align}
for $\mathrm{Re}(s)>0$.

Using the operator identity
\begin{align}
  (C - D)^{- 1} & =  C^{- 1} + C^{- 1} D (C - D)^{- 1}
   =  C^{- 1} + C^{- 1} D C^{- 1} + C^{- 1} D C^{- 1} D C^{- 1} + \ldots 
  \label{eq.opid}
\end{align}
with $C = s \1 - \frac{\mathcal{L}_0}{\varepsilon^2}$ and $D =
\frac{\mathcal{L}_{12}}{\varepsilon}$
we have
\begin{align}
\hspace{-60pt}   \mathfrak{L} \{ \mathcal{A}^{(p)} \} (s) & =  \left\langle \left( 
   \left( s
   \1 - \frac{\mathcal{L}_0}{\varepsilon^2} \right)^{- 1} +
   \left( s \1 - \frac{\mathcal{L}_0}{\varepsilon^2} \right)^{-
   1} \frac{\mathcal{L}_{12}}{\varepsilon} \left( s \1 -
   \frac{\mathcal{L}_0}{\varepsilon^2} \right)^{- 1} \right. \right. \nonumber \\
   & \left. \left. + \left( s \1 -
   \frac{\mathcal{L}_0}{\varepsilon^2} \right)^{- 1}
   \frac{\mathcal{L}_{12}}{\varepsilon} \left( s \1 -
   \frac{\mathcal{L}_0}{\varepsilon^2} \right)^{- 1}
   \frac{\mathcal{L}_{12}}{\varepsilon} \left( s \1 -
   \frac{\mathcal{L}_0}{\varepsilon^2} \right)^{- 1} + \ldots \right)
   z^{(p)} \right\rangle \nonumber\\
   & =  \frac{z^{(p)}}{s} + \frac{1}{s^2} \left\langle \left(
   \frac{\mathcal{L}_{12}}{\varepsilon} +
   \frac{\mathcal{L}_{12}}{\varepsilon} \left( s \1 -
   \frac{\mathcal{L}_0}{\varepsilon^2} \right)^{- 1}
   \frac{\mathcal{L}_{12}}{\varepsilon} \right. \right.  \label{eq.laplace1}\\
   & \left. \left. \qquad \qquad \qquad
    + \frac{\mathcal{L}_{12}}{\varepsilon} \left( s
   \1 - \frac{\mathcal{L}_0}{\varepsilon^2} \right)^{- 1}
   \frac{\mathcal{L}_{12}}{\varepsilon} \left( s \1 -
   \frac{\mathcal{L}_0}{\varepsilon^2} \right)^{- 1}
   \frac{\mathcal{L}_{12}}{\varepsilon} + \ldots \right) z^{(p)}
   \right\rangle , \nonumber
\end{align}
where we used in the last equality that $\mathcal{L}_0 A (x) =\mathcal{L}_0^* {\mu}_{x_0}^{(0)} = 0$.

We assume here that the resolvent is compact. This is the case on $L_p$-spaces if the generator includes diffusion. For the purely advective case of general deterministic multi-scale systems under consideration here, there is currently no theory available on what function spaces the resolvent is compact\footnote{For hyperbolic maps one has good statistical properties when considering anisotropic Banach spaces (see \cite{Baladi2018} and references therein).}. Under this assumption let us decompose
$s \1 - \frac{\mathcal{L}_0}{\varepsilon^2}$ into projectors $p_i$ on the eigenfunctions of $\mathcal{L}_0$, where $p_i \bullet = (r_i, \bullet) l_i$ and $p_0 \bullet = ({\mu}^{(0)}, \bullet) 1$, such that $\mathcal{L}_0 = \sum_{i=0}^\infty \lambda_i p_i$. The pairing $(a,b)$ is defined as $(a,b)=\int a(dy) b$. The eigenfunctions $l_i$ and $r_i$ satisfy $ (r_i, l_j) = \delta_{i,j}$ and correspond to the eigenvalue $\lambda_i$ with $\lambda_0 = 0$ \cite[Chapter 7]{MR832868}. 
Then
\begin{align*}
  s \1 - \frac{\mathcal{L}_0}{\varepsilon^2} & =  \sum_{i =
  0}^{\infty} \left( s - \frac{\lambda_i}{\varepsilon^2} \right) p_i = s p_0 +
  \left( s - \frac{\lambda_1}{\varepsilon^2} \right) p_1 + \left( s -
  \frac{\lambda_2}{\varepsilon^2} \right) p_2 + \ldots
\end{align*}
and therefore
\begin{align}
  \left( s \1 - \frac{\mathcal{L}_0}{\varepsilon^2} \right)^{-
  1} & =  \frac{1}{s} p_0 + \left( s - \frac{\lambda_1}{\varepsilon^2}
  \right)^{- 1} p_1 + \left( s - \frac{\lambda_2}{\varepsilon^2} \right)^{- 1}
  p_2 + \ldots 
  =  \frac{1}{s} p_0 + \left( s \1_{\perp} -
  \frac{\mathcal{L}_{0 \perp}}{\varepsilon^2} \right)^{- 1} \nonumber\\
  & =  \frac{1}{s} p_0 - \varepsilon^2 (\mathcal{L}_{0 \perp})^{- 1}
  - s \varepsilon^4 (\mathcal{L}_{0 \perp})^{- 2} - s^2 \varepsilon^6
  (\mathcal{L}_{0 \perp})^{- 3} - \ldots .  \label{eq.series}
\end{align}
where we define the operators restricted to the space orthogonal to the projection onto the invariant measure $\mu_{x_0}^{(0)}$ as $\mathcal{L}_{0 \perp} = \lambda_1 p_1 + \lambda_2 p_2 + \ldots$ and $\1_{\perp} = p_1 + p_2 + \ldots$. We note that formally 
\begin{align}
\mathcal{L}^{-1}_{0,\perp} = - \int_0^\infty \mathrm{d}\tau \left( e^{\tau \mathcal{L}_0} - p_0\right)
\label{e.L0inv}
\end{align}
and therefore formally $ \langle A \mathcal{L}^{-1}_{0\perp} f_0 \rangle = \langle A \mathcal{L}^{-1}_{0} f_0 \rangle$ and $ \langle f_0 \mathcal{L}^{-1}_{0\perp} A \rangle = \langle f_0 \mathcal{L}^{-1}_{0} A \rangle$ for any $A$ since $p_0 f_0 = \langle f_0 \rangle = 0$.

The expansion (\ref{eq.laplace1}) contains terms with both positive and negative powers of $s$. According to the residue theorem, only the residues contribute to the inverse Laplace transform. By the way it is defined in Eq. (\ref{eq.laplace1}), the function $\mathfrak{L} \{ \mathcal{A}^{(p)} \} (s)$ can only have poles at zero or at $\frac{\lambda_i}{\varepsilon^2}$. The poles at $\frac{\lambda_i}{\varepsilon^2}$ decay exponentially upon inverse Laplace transform (recall that $t/\eps^2 \to \infty$), and we therefore only need to consider the poles at zero.

\subsubsection{Combinatorial analysis of the expansion at each order in $\eps$ and $t$}
\label{sec.nonzero}
The expansion of the Laplace transform (\ref{eq.laplace0}) contains powers of $\eps$ and of $s$. To facilitate the combinatorial problem of accounting for the non-zero contributions at a specified order of $\eps$ and $s$ (or $t$, respectively) we introduce the following diagrammatic representation. 
We classify all terms in the expansion (\ref{eq.laplace1}) as sequences $S (a, b)$ with $a_i \in \mathbb{N}_0$, $b_j \in \{ 1, 2 \}$ of the form
\begin{align}
S (a, b) \assign \frac{1}{s^2} \langle B_{b_1} A_{a_1} B_{b_2} A_{a_2}
\ldots B_{b_l} A_{a_l} B_{b_{l + 1}} z^{(p)} \rangle , \label{eq.sequences}
\end{align}
where $A_0 = \frac{1}{s} p_0$, $A_1 = - \varepsilon^2 (\mathcal{L}_{0 \perp})^{- 1}$, $A_2 = - s \varepsilon^4 (\mathcal{L}_{0 \perp})^{- 2}$, etc. represent the terms in the series (\ref{eq.series}), and $B_1 =\mathcal{L}_1 / \varepsilon$ and $B_2 =\mathcal{L}_2$. We introduce the shorthand notation for a sequence of length $l$, $\left( \begin{array}{cccccc}
  & a_1 & \ldots &  & a_l & \\
  b_1 &  & \ldots & b_l &  & b_{l + 1}
\end{array} \right)$ for $S (a, b)$.
We also introduce a product $\otimes$ of two subsequences, with
\begin{align*}
 & \left( \begin{array}{cccccc}
  & a_1 & \ldots &  a_l & \\
  b_1 &  & \ldots &  & b_{l + 1}
\end{array} \right)
\otimes 
\left( \begin{array}{cccccc}
  & c_1 & \ldots &  c_l & \\
  d_1 &  & \ldots &  & d_{l + 1}
\end{array} \right) \\
= & \left( \begin{array}{cccccccccccc}
      & a_1 & \ldots &  a_l &           & 0 &     & c_1 & \ldots & c_l & \\
  b_1 &     & \ldots &      & b_{l + 1} &   & d_1 &     & \ldots &     & d_{l + 1}
\end{array} \right)  .
\end{align*}
For example
\begin{align*}
\frac{1}{s^2} \left\langle
\mathcal{L}_2 \frac{p_0}{s} \mathcal{L}_1 (- \varepsilon^2
(\mathcal{L}_{0 \perp})^{- 1}) \mathcal{L}_1 z^{(p)}
\right\rangle
= 
S ((0, 1), (2, 1, 1)) \\ 
=
\left( \begin{array}{ccccc}
  & 0 &  & 1 & \\
  2 &  & 1 &  & 1
\end{array} \right) = \left( \begin{array}{ccccc}
    \\
  2
\end{array} \right) \otimes \left( \begin{array}{ccccc}
    & 1 & \\
  1 &  & 1
\end{array} \right)   .
\end{align*}

The order of $\varepsilon$ of a given sequence is given by 
\[
O_\eps=
2 \sum_{i =
1}^l a_i + \sum_{j = 1}^{l + 1} (b_j - 2) = \left( 2 \sum_{i = 1}^l a_i +
\sum_{j = 1}^{l + 1} b_j \right) - 2 (l + 1) ,
\]
and the order of $s$ is given by
\[
O_s=
- 2 + \sum_{i = 1}^l (a_i - 1) ,
\]
 or, equivalently, the order of $t$ in the inverse Laplace transform is given by
 \[
 O_t = 1 - \sum_{i = 1}^l (a_i - 1) = 1 + l - \sum_{i =
1}^l a_i.
\]
The combined order in $t$ and $\varepsilon$ of the sequence is
\begin{align}
O_{\eps,t}
=\sum_{i = 1}^l a_i + \sum_{j = 1}^{l + 1} b_j - l - 1. \label{eq.comb-order}
\end{align}
Note that the combined order of a sequence does not change upon commutation of the subsequences in a product.
Finally, the sign of the sequence is given by $(- 1)^{\theta}$ where $\theta = \sum_{i = 1}^l \delta_{0,a_i}$ is the number of non-zero
elements among the $a_i$ where we used the Kronecker $\delta$.

We can readily identify certain sequences which do not contribute. The centering condition $\langle \mathcal{L}_1 \rangle = 0$ implies that sequences containing a subsequence $\left( \cdots \begin{array}{lll}
  0 &  & 0\\
  & 1 & 
\end{array} \cdots \right)$ vanish. For the same reason, sequences starting with $\left( \begin{array}{lll}
  & 0\\
1 & 
\end{array} \cdots \right)$ or ending in $\left( \cdots \begin{array}{lll}
0 &  \\
& 1 
\end{array} \right)$ vanish. Furthermore, sequences containing an insufficient number of derivatives with respect to $x$ do not contribute; if the number of derivatives is less than the order $p$, the sequence $S(a,b)$ averages to zero since $\delta_{x_0} \left( \partial_x^{n} (x - x_0)^p\right) = 0$ for $n<p$. Hence only those sequences with $l + 1 \geqslant p$ contribute since the operators $\mathcal{L}_1$ and $\mathcal{L}_2$ contain at most one $x$-derivative. This implies that for fixed moment order $p$ we have
\begin{align}
l \geq l_{\textrm{min}} = p - 1 \label{eq.lmin}
\end{align}
Recall that moments are rescaled by $t^{-\frac{p}{2}}$ (cf. \eqref{e.exp0}). We therefore require
\begin{align}
O_{\eps,t} - \lfloor \frac{p}{2} \rfloor \leq 1
\end{align}
to retain only terms that contribute to first order after normalization. This implies
\begin{align}
\sum_{i = 1}^l a_i + \sum_{j = 1}^{l + 1} b_j - l - 1 & \leqslant q + 1  .
\label{eq.order}
\end{align}
We now identify all those sequences which are compatible with the constraints \eqref{eq.order} and \eqref{eq.lmin}. For sake of exposition, we treat even and odd moments separately.
We first determine the possible sequences of different lengths $l$ for even moments $m^{(2 q)}$ where $l\geqslant 2 q - 1$.


\paragraph{Case $l = 2 q - 1$}
Equation~\eqref{eq.order} implies $\sum_{i = 1}^l a_i + \sum_{j = 1}^{l + 1} b_j \leqslant 3 q + 1$. The sequences of length $2 q - 1$ have in total $2 l + 1 = 4 q - 1$ elements. Each sequence must therefore contain at least $q-2$ zero elements. Since sequences with subsequence 
$\left( \cdots \begin{array}{lll}
  0 &  & 0\\
  & 1 & 
\end{array} \cdots \right)$ average to zero, there are 6 possible types of
sequence that fulfill the inequality, namely
\begin{eqnarray}
  \left(\begin{array}{ccc}
      & 1 &  \\
    1 &   & 1
  \end{array}\right)^{ \mbox{\normalsize $\otimes q$}} \qquad \;\;\; & \qquad & \begin{array}{l}
    \mbox{permutations: $1$} \\
    \mbox{$O_{\eps,t} =q$}
  \end{array}
  \label{seq.hom1}\\
  \left(\begin{array}{c}
      \\
    2
  \end{array}\right)^{ \mbox{\normalsize $\otimes 2$}} \otimes
  \left(\begin{array}{ccc}
      & 1 &  \\
    1 &   & 1
  \end{array}\right)^{ \mbox{\normalsize $\otimes (q-1)$}}
  & \qquad &  \begin{array}{l}
    \mbox{permutations: ${{q+1}\choose{2}}$}\\
    \mbox{$O_{\eps,t} =q + 1 $}
  \end{array}
  \label{seq.hom2}\\
   \left(\begin{array}{ccccccc}
      & 1 &   & 1 &   & 1 &   \\
    1 &   & 1 &   & 1 &   & 1
  \end{array}\right) \otimes 
  \left(\begin{array}{ccc}
      & 1 &  \\
    1 &   & 1
  \end{array}\right)^{ \mbox{\normalsize $\otimes (q-2)$}}
  & \qquad & \begin{array}{l}
      \mbox{permutations: ${{q-1}\choose{1}}$}\\
    \mbox{$O_{\eps,t} =q + 1 $}
  \end{array}
  \label{seq.sf1}\\
    \left(\begin{array}{ccc}
      & 1 &   \\
    2 &   & 1
  \end{array}\right) \otimes
  \left(\begin{array}{ccc}
      & 1 &   \\
    1 &   & 1
  \end{array}\right) ^{ \mbox{\normalsize $\otimes (q-1)$}}
  & \qquad &
  \begin{array}{l}
    \mbox{permutations: ${{q}\choose{1}}$} \\
    \mbox{$O_{\eps,t} =q + 1 $}
  \end{array}
  \label{seq.sf2}\\
    \left(\begin{array}{ccc}
      & 2 &   \\
    1 &   & 1
  \end{array}\right) \otimes
  \left(\begin{array}{ccc}
      & 1 &   \\
    1 &   & 1
  \end{array}\right) ^{ \mbox{\normalsize $\otimes (q-1)$}}
  & \qquad &
  \begin{array}{l}
    \mbox{permutations: ${{q}\choose{1}}$} \\
    \mbox{$O_{\eps,t} =q + 1 $}
  \end{array}
  \label{seq.sf3}\\
    \left(\begin{array}{ccc}
      & 1 &   \\
    1 &   & 2
  \end{array}\right) \otimes
  \left(\begin{array}{ccc}
      & 1 &   \\
    1 &   & 1
  \end{array}\right) ^{ \mbox{\normalsize $\otimes (q-1)$}}
  & \qquad &
  \begin{array}{l}
    \mbox{permutations: ${{q}\choose{1}}$} \\
    \mbox{$O_{\eps,t} =q + 1 $}
  \end{array}
  \qquad \qquad
  \label{seq.sf4}
\end{eqnarray}
with all possible permutations of subsequences. For $q=2$, we can verify by simple enumeration that these are all the sequences possible if we restrict to $b_j<3$. All other possible nonzero terms with $b_j>3$ can be obtained by substituting a $b_j$ in one of these sequences, but this operation increases the order. For $q>2$ the only way to extend a series without increasing the order is to append an element $\left(\begin{array}{ccc}
& 1 &   \\
1 &   & 1
\end{array}\right)$.

\paragraph{Case $l = 2 q$}
Equation~\eqref{eq.order} implies $\sum_{i = 1}^l a_i + \sum_{j = 1}^{l + 1} b_j \leqslant 3 q + 2$. The sequences of length $2 q$ have in total $2 l + 1 = 4 q + 1$ elements. Each sequence must therefore contain at least $q-1$ zero elements. There is only 1 possible type of sequence that fulfills the inequality, namely
\begin{eqnarray}
  \hspace{-50pt}\left(\begin{array}{ccccc}
      & 1 &   & 1 &   \\
    1 &   & 1 &   & 1
  \end{array}\right) \otimes
    \left(\begin{array}{ccc}
      & 1 &   \\
    1 &   & 1 
  \end{array}\right) ^{ \mbox{\normalsize $\otimes (q-1)$}}
   & \qquad   &
  \begin{array}{l}
    \mbox{permutations: ${{q}\choose{1}}$} \\
    \mbox{$O_{\eps,t} =q + 1 $}
  \end{array} \label{seq.sf5}
\end{eqnarray}
with all possible permutations of subsequences.

\paragraph{Case $l = 2 q + 1$}
Equation~\eqref{eq.order} implies $\sum_{i = 1}^l a_i + \sum_{j = 1}^{l + 1} b_j \leqslant 3 q + 3$. The sequences of length $2 q + 1$ have in total $2 l + 1 = 4 q + 3$ elements. Each sequence must therefore contain at least $q$ zero elements. There is only 1 possible type of sequence that fulfills the inequality, namely
\begin{eqnarray}
  \left(\begin{array}{ccc}
    & 1 &  \\
    1 &  & 1
  \end{array}\right)^{ \mbox{\normalsize $\otimes (q + 1)$}}  & \qquad &
  \begin{array}{l}
  \mbox{permutations: $1$} \\
  \mbox{$O_{\eps,t} =q + 1 $}
  \end{array}
  \label{seq.hom3}
\end{eqnarray}
with all possible permutations of subsequences.

\paragraph{Case $l \geqslant 2 q + 2$}
Nonzero sequences of length $2 q + 2$ are of at least order $q + 2$, so they do not contribute.\\

For odd moments $m^{(2 q + 1)}$, the inequality  (\ref{eq.order}) together with the condition on the number of derivatives $l\geqslant 2q$ implies that only
sequences (\ref{seq.sf5})-(\ref{seq.hom3}) contribute.

\subsubsection{Calculation of the moments $m^{(p)}$}
\label{sec.calcmoments}
We have now derived all nontrivial sequences contributing to the moments $m^{(p)}$. We can obtain explicit expressions for the moments by calculating for each sequence the expectation value as prescribed in Eq. \eqref{eq.sequences}. In the process, subsequences become functions of $x$ due to the projection $p_0$. Finally, we work out explicitly the $x$-derivatives in the sequence. For example, for the sequence $S((0,1),(2,1,1))$ with $p=3$, we have that, as described above,
\begin{align*}
\mathfrak{L}^{-1}S((0,1),(2,1,1)) &= - \varepsilon^2 t^2 \left\langle
\mathcal{L}_2 p_0 \mathcal{L}_1 
\mathcal{L}_{0 \perp}^{- 1} \mathcal{L}_1 (x-x_0)^3
\right\rangle \\
&= - 3 \, \varepsilon^2 t^2  \left\langle f_1(x,y) \right\rangle \left\langle f_0(x,y) 
\mathcal{L}_{0\perp}^{-1} f_0(x,y) \right\rangle .
\end{align*}

Recalling the multi-index notation to denote the expansions in $\eps$ and $t$ with ${\bmalpha}=(\alpha_\eps,\alpha_t)$ and with $|{\bmalpha}|=\alpha_\eps+\alpha_t$ being the combined order of the contribution, we obtain, in this way, from sequence (\ref{seq.hom1}) the  leading homogenized term ($|{\bmalpha}|=0$) in the $2 q$-th scaled moment $m^{(2 q)}$
\begin{align}
  m_0^{(2 q)} =  \frac{1}{t^q} \frac{t^q}{q!} (2 q)! \left( \frac{\sigma^{2}}{2} \right)^q, \label{eq.m02q}
\end{align}
where $\sigma^2 = - 2 \langle f_0 \mathcal{L}_{0\perp}^{-1} f_0 \rangle$ is the homogenized diffusion coefficient \eqref{eq.contGK}. The first factor $1/t^q$ is the normalization factor of the moment, the term $t^q/q!$ comes from the inverse Laplace transform of $1/s^{q+1}$ and $(2q)!$ comes from the $2q$-th derivative of $(x-x_0)^{2q}$. Finally, since only terms with $2q$ derivatives $\partial_x$ yield nonzero contributions, only the term with $f_0$ in $\mathcal{L}_1$ contributes, resulting in $\left\langle f_0 e^{t \mathcal{L}_0} f_0 \right\rangle ^q = \left( \sigma^{2}/{2} \right)^q$.

Similar combinatorial accounting distills from the sequences (\ref{seq.hom2}) and (\ref{seq.hom3}) the homogenized first correction term ($\alpha_\eps = 0$, $|{\bmalpha}| = 1$) in the $2 q$-th scaled moment $m^{(2 q)}$
\begin{align*}
t \, m_{0,1}^{(2 q)} & =  t \, \frac{(2 q) !}{(p + 1) !} \left( 
  {{q+1}\choose{2}} F^2 \left( \frac{\sigma^2}{2^{}} \right)^{q - 1} + F \left(
  \frac{\sigma^2}{2} \right)^{(q - 1)} \sigma \partial_x \sigma \left( q^3 +
  \frac{q^2}{2} - \frac{q}{2} \right) \right. \\
  &\left.
  \phantom{t \frac{(2 p) !}{(p + 1) !}} + q (q + 1) \left( \frac{\sigma^2}{2}
  \right)^q \partial_x F + \frac{\sigma^{(2 q + 1)}}{2^{(q + 1)}} \partial^2_x
  \sigma \left( \frac{4}{3} q^3 + q^2 - \frac{q}{3} \right) \right. \\
  &\left.
  \phantom{t \frac{(2 p) !}{(p + 1) !}} + \frac{\sigma^{2 q}}{2^{q + 1}}
  (\partial_x \sigma)^2  \frac{q}{3} (q + 1) (6 q^2 - 2 q - 1) \right) ,
\end{align*}
where $F = \langle f_1 \rangle - \langle (g_1 \partial_y) \mathcal{L}_{0\perp}^{-1} f_0 \rangle - \langle f_0 \mathcal{L}_{0\perp}^{-1} \partial_x f_0 \rangle$ is the homogenized drift coefficient \eqref{eq.contF}.

From sequences (\ref{seq.sf1})-(\ref{seq.sf4}) we obtain the non-homogenized corrections ($\alpha_\eps > 0$, $|{\bmalpha}| = 1$) to the $2 q$-th scaled moment with $2 q$ $x$-derivatives
\begin{align*}
\eps\,  {\tilde{m}}_{1,0}^{(2 q)} +\frac{\eps^2}{t}\,m_{2,-1}^{(2 q)} & =  \frac{(2 q) !}{t^q} \left( \frac{\varepsilon
  t^q}{q!} q (q - 1) \langle f_0 \mathcal{L}_{0\perp}^{-1} f_0
  \mathcal{L}_{0\perp}^{-1} f_0 \rangle \langle f_1 \rangle \left( - \left\langle f_0
  \mathcal{L}_{0\perp}^{-1} f_0 \right\rangle \right)^{(q - 2)} \right. \\
  & \left.
  \phantom{\frac{(2 p) !}{t^p} (} + \frac{\varepsilon t^q}{q!} q
  \left( - \langle f_1 \mathcal{L}_{0\perp}^{-1} f_0 \rangle \right) (-
  \langle f_0 \mathcal{L}_{0\perp}^{-1} f_0 \rangle)^{q - 1} \right. \\
  & \left.
  \phantom{\frac{(2 p) !}{t^p} (} + \frac{\varepsilon t^q}{q!} q
  \left( - \left\langle f_0 \mathcal{L}_{0\perp}^{-1} f_1
  \right\rangle \right) (- \langle f_0 \mathcal{L}_{0\perp}^{-1} f_0
  \rangle)^{(q - 1)} \right. \\
  & \left.
  \phantom{\frac{(2 p) !}{t^p} (} + \frac{\varepsilon^2 t^{q -
  1}}{(q - 1) !} q \left( - \left\langle f_0 \mathcal{L}_{0\perp}^{- 2} f_0
  \right\rangle \right) (- \langle f_0 \mathcal{L}_{0\perp}^{-1} f_0
  \rangle)^{q - 1} \right. \\
  & \left.
  \phantom{\frac{(2 p) !}{t^p} (} + \frac{\varepsilon^2 t^{q -
  1}}{(q - 1) !} (q - 1) \left( - \left\langle f_0 \mathcal{L}_{0\perp}^{-1} f_0 \mathcal{L}_{0\perp}^{- 1} f_0 \mathcal{L}_{0\perp}^{-1} f_0
  \right\rangle \right) (- \langle f_0 \mathcal{L}_{0\perp}^{-1} f_0
  \rangle)^{q - 2} \right. \\
  & \left.
  \phantom{\frac{(2 p) !}{t^p} (} + \frac{\varepsilon^2 t^{q -
  1}}{(q - 1) !} 
  {{q-1}\choose{2}}
  \left\langle f_0 \mathcal{L}_{0\perp}^{-1} f_0 \mathcal{L}_{0\perp}^{-1} f_0 \right\rangle^2  (-
  \langle f_0 \mathcal{L}_{0\perp}^{-1} f_0 \rangle)^{q - 3} \right).
\end{align*}
From sequence (\ref{seq.sf5}) we obtain the non-homogenized corrections ($\alpha_\eps > 0$, $|{\bmalpha}| = 1$) to the $2 q$-th moment with $2 q + 1$ $x$-derivatives
\begin{align*}
\eps\,  {\tilde{\tilde{m}}}_{1,0}^{(2 q)} & =   (2 q) ! \frac{\varepsilon}{q!}  \left( q
  \left( \langle (g_1 \partial_y) \mathcal{L}_{0\perp}^{- 1} f_0
  \mathcal{L}_{0\perp}^{-1} f_0 \rangle (- \langle f_0
  \mathcal{L}_{0\perp}^{-1} f_0 \rangle)^{q - 1} \right. \right. \\
  & \left. \left.
  + \langle^{} f_0 \mathcal{L}_{0\perp}^{-1} (g_1 \partial_y)
  \mathcal{L}_{0\perp}^{-1} f_0 \rangle (- \langle f_0 \mathcal{L}_{0\perp}^{-1}
  f_0 \rangle)^{q - 1} \right. \right. \\
  & \left. \left.
  + (q - 1) \left\langle f_0 \mathcal{L}_{0\perp}^{-1} f_0 \mathcal{L}_{0\perp}^{-1} f_0 \right\rangle (- \langle (g_1
  \partial_y) \mathcal{L}_{0\perp}^{-1} f_0 \rangle (- \langle f_0
  \mathcal{L}_{0\perp}^{-1} f_0 \rangle)^{q - 2}) \right) \right. \\
  & \left.
  + q (q - 1) \left( - \left\langle f_0 \mathcal{L}_{0\perp}^{-1} f_0
  \right\rangle \right)^{q - 1} \langle \partial_x f_0
  \mathcal{L}_{0\perp}^{-1} f_0 \mathcal{L}_{0\perp}^{-1} f_0 \rangle \right. \\
  & \left.
  + q^2 \left( - \langle f_0 \mathcal{L}_{0\perp}^{-1} f_0 \rangle
  \right)^{q - 1} \langle f_0 \mathcal{L}_{0\perp}^{-1} \partial_x f_0
  \mathcal{L}_{0\perp}^{-1} f_0 \rangle \right. \\
  & \left.
  + q (q + 1) (- \langle f_0 \mathcal{L}_{0\perp}^{-1} f_0 \rangle)^{q -
  1} \left\langle f_0 \mathcal{L}_{0\perp}^{-1} f_0
  \mathcal{L}_{0\perp}^{-1} \partial_x f_0 \right\rangle \right. \\
  & \left.
  + q \left( q^2 - \frac{3}{2} q + \frac{1}{2} \right) (- \langle \partial_x
  f_0 \mathcal{L}_{0\perp}^{-1} f_0 \rangle)  \left( - \left\langle f_0
  \mathcal{L}_{0\perp}^{-1} f_0 \right\rangle \right)^{q - 2} \langle f_0
  \mathcal{L}_{0\perp}^{-1} f_0 \mathcal{L}_{0\perp}^{-1} f_0 \rangle \right. \\
  & \left.
  + q \left( q^2 - \frac{1}{2} q - \frac{1}{2} \right) \left( - \left\langle
  f_0 \mathcal{L}_{0\perp}^{-1} \partial_x f_0 \right\rangle \right)  (-
  \langle f_0 \mathcal{L}_{0\perp}^{-1} f_0 \rangle)^{q - 2} \langle f_0
  \mathcal{L}_{0\perp}^{-1} f_0 \mathcal{L}_{0\perp}^{-1} f_0 \rangle
  \right) ,
\end{align*}
where ${{{m}}}_{1,0}^{(2 q)}={\tilde{{m}}}_{1,0}^{(2 q)}+{\tilde{\tilde{m}}}_{1,0}^{(2 q)}$. From sequence (\ref{seq.hom3}) we obtain the leading homogenized terms ($\alpha_\eps = 0$, $|{\bmalpha}| = \frac{1}{2}$) in the $\left(2 q + 1\right)$-th moment
\begin{align*}
\sqrt{t}\,  m_{0,{\tiny{\frac{1}{2}}}}^{(2 q + 1)} & =  \sqrt{t}  \frac{(2 q + 1) !}{q!} \frac{1}{2^q} 
  (F \sigma^{2 q} + q \sigma^{2 q + 1} \partial_x \sigma).
\end{align*}
From sequence (\ref{seq.sf5}) we obtain the non-homogenized correction
terms ($\alpha_\eps > 0$, $|{\bmalpha}| = \frac{1}{2}$) in the $(2 q + 1)$-th moment
\begin{align*}
\frac{\eps}{\sqrt{t}}\, m_{1, {\tiny{-\frac{1}{2}}}}^{(2 q + 1)} & =  \frac{(2 q + 1) !}{t^{(2 q + 1) / 2}} \varepsilon
  \frac{t^q}{q!} q \langle f_0 \mathcal{L}_{0\perp}^{-1} f_0
  \mathcal{L}_{0\perp}^{-1} f_0 \rangle (- \langle f_0
  \mathcal{L}_{0\perp}^{-1} f_0 \rangle)^{q - 1}.
\end{align*}
Summarizing for $q=0$ and $q=1$ we obtain the desired expressions for the first four cumulants \eqref{e.m1}, \eqref{e.m2}, \eqref{e.m3} and \eqref{e.m4}. 
We have until now only performed the average with respect to the leading order contribution $\mu^{(0)}_{x_0}$. The following Lemma shows that this is sufficient
\begin{lemma}
The asymptotic series of cumulants $c^{(p)}$ in $\eps$ and $t$ up to the combined order $O_{\eps,t}=\frac{3}{2}$ involves averages over $\mu_{x_0}^{(0)}$ only and the linear response term $\mu_{x_0}^{(1)}$ does not contribute.
\label{Lemma_LinResp}
\end{lemma}
\begin{proof}
Analogously to \eqref{eq.laplace0}, the Laplace transform of 
\begin{align*}
\mathcal{B}^{(p)} (t)  = \varepsilon 
{\mu}_{x_0}^{(1)} 
\left( e^{t \left(
\frac{\mathcal{L}_0}{\varepsilon^2} +
\frac{\mathcal{L}_{12}}{\varepsilon} \right)} z^{(p)} (x) \right)
\end{align*}
can be expanded as
\begin{align*}
  \mathfrak{L} \{ \mathcal{B}^{(p)} \} (s) & =  \varepsilon
  {\mu}^{(1)}_{x_0} \left( \left( \left( s \1 -
  \frac{\mathcal{L}_0}{\varepsilon^2} \right)^{- 1} + \left( s
  \1 - \frac{\mathcal{L}_0}{\varepsilon^2} \right)^{- 1}
  \frac{\mathcal{L}_{12}}{\varepsilon} \left( s \1 -
  \frac{\mathcal{L}_0}{\varepsilon^2} \right)^{- 1} \right. \right. \\
  & \left. \left.
  \hspace*{\fill} + \left( s \1 -
  \frac{\mathcal{L}_0}{\varepsilon^2} \right)^{- 1}
  \frac{\mathcal{L}_{12}}{\varepsilon} \left( s \1 -
  \frac{\mathcal{L}_0}{\varepsilon^2} \right)^{- 1}
  \frac{\mathcal{L}_{12}}{\varepsilon} \left( s \1 -
  \frac{\mathcal{L}_0}{\varepsilon^2} \right)^{- 1} + \ldots \right)
  z^{(p)} \right).
\end{align*}
Using our diagrammatic representation, the sequences are encoded as
\begin{align*}
S(a,b)=\left(\begin{array}{ccccc}
  a_1 &  & \ldots &  & a_l\\
  & b_1 & \ldots & b_{l - 1} & 
\end{array}\right).
\end{align*}
The order of $\varepsilon$ of a given sequence is now $O_\eps=2 \sum_{i = 1}^l a_i + \sum_{j = 1}^{l - 1} (b_j - 2) + 1 = \left( 2 \sum_{i = 1}^l a_i + \sum_{j = 1}^{l - 1} b_j \right) - 2 (l - 1) + 1$ and the order of $t$ is $ O_t=- 1 - \sum_{i = 1}^l (a_i - 1) = - 1 + l - \sum_{i = 1}^l a_i$. The combined order in $t$ and $\varepsilon$ of the sequence is then $O_{\eps,t}=\sum_{i = 1}^l a_i + \sum_{j = 1}^{l + 1} b_j - l + 2$. The only sequence that can contribute in this case is
\begin{align*}
\left(\begin{array}{cccc}
    0  & & 1 &  \\
    & 1 &   & 1
  \end{array}\right)
  \otimes
  \left(\begin{array}{cccc}
        & 1 &  \\
      1 &   & 1
  \end{array}\right)^{ \mbox{\normalsize $\otimes q-2$}}
  \otimes
  \left(\begin{array}{cccc}
      & 1 &   & 0\\
    1 &   & 1 & 
  \end{array}\right) 
\end{align*}
at order $O_{\eps,t}=q+1$. Because of the left-most projection $p_0$ in the sequence, such terms are functions of $x$ only. Since ${\mu}_{x_0}^{(1)}= \lim_{\eps \rightarrow 0} (\mu^{(\eps)}-\mu_{x_0}^{(0)})/\eps$ is a difference of two normalized measures, applying ${\mu}_{x_0}^{(1)}$ to a constant in $y$ yields $0$. Therefore the cumulants up to $\Ord (\eps)$ only contain averages over the fast invariant measure $\mu_{x_0}^{(0)}$ and the average with respect to the linear response $\mu_{x_0}^{(1)}$ -- the second term in (\ref{eq.momentp2}) -- vanishes.
\end{proof}
\begin{rmk}
The linear response term $\mu^{(1)}_{x_0}$ may become relevant for higher orders $O_{\eps,t}\geqslant 2$.
\end{rmk}

Finally, we prove that higher order cumulants $c^{(p)}$ with $p\geqslant 5$ do not contribute at $\mathcal{O} (\varepsilon)$, and we can be content finding expressions for the first 4 cumulants to determine the Edgeworth expression. 
\begin{lemma}
$c^{(p)}=\mathcal{O}(\eps^\frac{3}{2})$ for $p\geqslant 5$.
\label{Lemma_HigherOrder}
\end{lemma}
\begin{proof}
Using the recursion formula for cumulants
\begin{align*}
  c^{(p)} & =  m^{(p)} - \sum_{m = 1}^{p - 1} 
  {{p - 1}\choose{m - 1}}
  m^{(p - m)} c^{(m)} ,
\end{align*}
we have by recursion at order $\Ord(1)$ for $q \geqslant 2$ that
\begin{align*}
  c^{(2 q)} & =  m_0^{(2 q)} - 
  {{2q - 1 }\choose{1}} m_0^{(2 q - 2)} c_0^{(2)}
  +\mathcal{O} (\eps) = c^{(2 q)}_0 +\mathcal{O} (\eps) ,
\end{align*}
where, due to \eqref{eq.m02q},
\begin{align*}
  c^{(2 q)}_0 & =  \frac{(2 q) !}{q!} \frac{\sigma^{2 q}}{2^q} - \frac{(2 q -
  1) !}{(2 q - 2) !} \frac{(2 q - 2) !}{(q - 1) !} \frac{\sigma^{2 q -
  2}}{2^{q - 1}} \sigma^2 = 0.
\end{align*}
Similarly, for the combined order $1$ contribution to the even cumulant $ \delta^{(1)} c^{(2 q)} = c^{(2 q)}_{1,0} + c^{(2 q)}_{0,1} +  c^{(2 q)}_{2,-1}$, we find that for $q \geqslant 3$
\begin{align}
  \delta^{(1)} c^{(2 q)} & =  \delta^{(1)} m^{(2 q)}
  - {{2q - 1}\choose{0}} \left( m_{0,\frac{1}{2}}^{(2 q - 1)} + m_{1,-\frac{1}{2}}^{(2 q - 1)} \right) c_0^{(1)} \nonumber \\
  &- {{2q - 1}\choose{1}}  (m_0^{(2 q - 2)} \delta^{(1)} c^{(2)} + \delta^{(1)} m^{(2 q - 2)}  c_0^{(2)}) \nonumber \\
  &- {{2q - 1}\choose{2}} \left( m_{0,\frac{1}{2}}^{(2 q - 3)} + m_{1,-\frac{1}{2}}^{(2 q - 3)} \right) c_0^{(3)} 
  - {{2q - 1}\choose{3}} m_0^{(2  q - 4)} \delta^{(1)} c_{}^{(4)} = 0 , 
  \label{eq.homdeltacum}
\end{align}
where $\delta^{(1)} m^{(2 q)} = m^{(2 q)}_{1,0} +  m^{(2 q)}_{0,1} +  m^{(2 q)}_{2,-1}$.
For the odd cumulants $c^{(2 q + 1)}$ we find recursively that the $\mathcal{O} \left( \sqrt{\eps} \right)$ contribution for $q \geqslant 2$ is
\begin{align*}
  c_{0,\frac{1}{2}}^{(2q+1)} + c_{1,{\tiny{-\frac{1}{2}}}}^{(2 q + 1)} & =  \left( m_{0,\frac{1}{2}}^{2 q + 1} + m_{1,-\frac{1}{2}}^{2 q + 1} \right) - m_0^{2 q} c_{0,\frac{1}{2}}^{(1)} 
  - (2 q) \left( m_{0,\frac{1}{2}}^{(2 q + 1)} + m_{1,-\frac{1}{2}}^{(2 q + 1)} \right) c_0^{(2)} \\
  &- {{2q}\choose{2}} m_0^{(2 q)} \left( c_{0,\frac{1}{2}}^{(3)} + c_{1,-\frac{1}{2}}^{(3)} \right) 
  = 0.
\end{align*}
\end{proof}



\begin{thebibliography}{99}

\bibitem{kamerlin_coarse-grained_2011}
Kamerlin SCL, Vicatos S, Dryga A, Warshel A. 2011  Coarse-{{Grained}}
  ({{Multiscale}}) {{Simulations}} in {{Studies}} of {{Biophysical}} and
  {{Chemical Systems}}. {\em Annual Review of Physical Chemistry} \textbf{62},
  41--64.

\bibitem{imkeller_stochastic_2001}
Imkeller P, {von Storch} JS. 2001 {\em Stochastic Climate Models}.
{Birkh{\"a}user}.

\bibitem{GivonEtAl04}
Givon D, Kupferman R, Stuart A. 2004  Extracting macroscopic dynamics: Model
  problems and algorithms. {\em Nonlinearity} \textbf{17}, R55--127.

\bibitem{Pavliotis08}
Pavliotis G, Stuart A. 2008 {\em Multiscale Methods Averaging and
  Homogenization}.
Texts in Applied Mathematics {\bf 53}, Springer.

\bibitem{Khasminsky66}
Khasminsky RZ. 1966  On stochastic processes defined by differential equations
  with a small parameter. {\em Theory of Probability and its Applications}
  \textbf{11}, 211--228.

\bibitem{Kurtz73}
Kurtz TG. 1973  A limit theorem for perturbed operator semigroups with
  applications to random evolutions. {\em Journal of Functional Analysis}
  \textbf{12}, 55--67.

\bibitem{Papanicolaou76}
Papanicolaou GC. 1976  Some probabilistic problems and methods in singular
  perturbations. {\em Rocky Mountain Journal of Mathematics} \textbf{6},
  653--674.

\bibitem{Beck90}
Beck C. 1990  Brownian motion from deterministic dynamics. {\em Phys. A}
  \textbf{169}, 324--336.

\bibitem{JustEtAl01}
Just W, Kantz H, R{\"o}denbeck C, Helm M. 2001  Stochastic modelling: replacing
  fast degrees of freedom by noise. {\em J.~Phys.~A} \textbf{34}, 3199--3213.

\bibitem{MelbourneStuart11}
Melbourne I, Stuart A. 2011  A note on diffusion limits of chaotic skew-product
  flows. {\em Nonlinearity} \textbf{24}, 1361--1367.

\bibitem{GottwaldMelbourne13c}
Gottwald GA, Melbourne I. 2013  Homogenization for deterministic maps and
  multiplicative noise. {\em Proceedings of the Royal Society A: Mathematical,
  Physical and Engineering Science} \textbf{469}.

\bibitem{KellyMelbourne17}
Kelly D, Melbourne I. 2017  Deterministic homogenization for fast--slow systems
  with chaotic noise. {\em Journal of Functional Analysis} \textbf{272},
  4063--4102.

\bibitem{Dolgopyat04}
Dolgopyat D. 2004  Limit theorems for partially hyperbolic systems. {\em Trans.
  Amer. Math. Soc.} \textbf{356}, 1637--1689).

\bibitem{DeSimoiLiverani15}
{De Simoi} J, Liverani C. 2015  The Martingale Approach after {V}aradhan and
  {D}olgopyat. In {\em Hyperbolic dynamics, fluctuations and large deviations}
  vol.~89{\em Proc. Sympos. Pure Math.} pp. 311--339. Amer. Math. Soc.,
  Providence, RI.

\bibitem{DeSimoiLiverani16}
{De Simoi} J, Liverani C. 2016  Statistical properties of mostly contracting
  fast-slow partially hyperbolic systems. {\em Invent. Math.} \textbf{206},
  147--227.

\bibitem{GearKevrekidis03}
Gear C, Kevrekidis I. 2003  Projective methods for differential equations. {\em
  SIAM J. Sci. Comp.} \textbf{24}, 1091--1106.

\bibitem{KevrekidisGearEtAl03}
Kevrekidis IG, Gear CW, Hyman JM, Panagiotis GK, Runborg O, Theodoropoulos C.
  2003  Equation-free, coarse-grained multiscale computation: {E}nabling
  microscopic simulators to perform system-level analysis. {\em Comm. Math.
  Sci.} \textbf{1}, 715--762.

\bibitem{E03}
E W. 2003  Analysis of the heterogeneous multiscale method for ordinary
  differential equations. {\em Comm. Math. Sci.} \textbf{1}, 423--436.

\bibitem{EEtAl07}
E W, Engquist B, Li X, Ren W, Vanden-Eijnden E. 2007  Heterogeneous multiscale
  methods: {A} review. {\em Comm. Comp. Phys.} \textbf{2}, 367--450.

\bibitem{MTV99}
Majda AJ, Timofeyev I, {Vanden-Eijnden} E. 1999  Models for stochastic climate
  prediction. {\em Proceedings of the National Academy of Sciences}
  \textbf{96}, 14687--14691.

\bibitem{Majdaetal01}
Majda AJ, Timofeyev I, {Vanden-Eijnden} E. 2001  A mathematical framework for
  stochastic climate models. {\em Communications on Pure and Applied
  Mathematics} \textbf{54}, 891--974.

\bibitem{MTV02}
Majda AJ, Timofeyev I, Vanden-Eijnden E. 2002  A priori tests of a stochastic
  mode reduction strategy. {\em Phys. D} \textbf{170}, 206--252.

\bibitem{Majda03}
Majda AJ, Timofeyev I, Vanden-Eijnden E. 2003  Systematic strategies for
  stochastic mode reduction in climate. {\em Journal of the Atmospheric
  Sciences} \textbf{60}, 1705--1722.

\bibitem{monahan_stochastic_2011}
Monahan AH, Culina J. 2011  Stochastic {{Averaging}} of {{Idealized Climate
  Models}}. {\em Journal of Climate} \textbf{24}, 3068--3088.

\bibitem{culina_stochastic_2011}
Culina J, Kravtsov S, Monahan AH. 2011  Stochastic {{Parameterization Schemes}}
  for {{Use}} in {{Realistic Climate Models}}. {\em Journal of the Atmospheric
  Sciences} \textbf{68}, 284--299.

\bibitem{GottwaldEtAl17}
Gottwald G, Crommelin D, Franzke C. 2017  Ensemble-based Atmospheric Data
  Assimilation. In Franzke CLE, O'Kane TJ, editors, {\em Nonlinear and
  Stochastic Climate Dynamics} pp. 209--240. Cambridge: Cambridge University
  Press.

\bibitem{BhattacharyaRanga}
Bhattacharya RN, Rao RR. 2010 {\em Normal Approximation and Asymptotic
  Expansions} vol.~64{\em Classics in Applied Mathematics}.
Society for Industrial and Applied Mathematics (SIAM), Philadelphia, PA.

\bibitem{fellerbook}
Feller W. 1957 {\em An Introduction to Probability Theory and Its
  Applications}.
New York: {Wiley} 2d ed edition.

\bibitem{simoi_fastslow_2017}
de~Simoi J, Liverani C, Poquet C, Volk D. 2017  Fast\textendash{}{{Slow
  Partially Hyperbolic Systems Versus Freidlin}}\textendash{}{{Wentzell Random
  Systems}}. {\em Journal of Statistical Physics} \textbf{166}, 650--679.

\bibitem{Dolgopyat05}
Dolgopyat D. 2005  Averaging and invariant measures. {\em Mosc. Math. J.}
  \textbf{5}, 537--576, 742.

\bibitem{Gouezel04}
Gou{\"e}zel S. 2004  Central limit theorem and stable laws for intermittent
  maps. {\em Probability Theory and Related Fields} \textbf{128}, 82--122.

\bibitem{cinlarbook}
{{\c C}{\i}nlar} E. 2011 {\em Probability and Stochastics}.
Number 261 in Graduate Texts in Mathematics. New York ; London: {Springer}.

\bibitem{MelbourneNicol05}
Melbourne I, Nicol M. 2005  Almost sure invariance principle for nonuniformly
  hyperbolic systems. {\em Commun. Math. Phys.} \textbf{260}, 131--146.

\bibitem{MelbourneNicol08}
Melbourne I, Nicol M. 2008  Large deviations for nonuniformly hyperbolic
  systems. {\em Trans. Amer. Math. Soc.} \textbf{360}, 6661--6676.

\bibitem{MelbourneNicol09}
Melbourne I, Nicol M. 2009  A vector-valued almost sure invariance principle
  for hyperbolic dynamical systems. {\em Annals of Probability} \textbf{37},
  478--505.

\bibitem{Gouezel10}
Gou{\"e}zel S. 2010  Almost sure invariance principle for dynamical systems by
  spectral methods. {\em Ann. Probability} \textbf{38}, 1639--1671.

\bibitem{GoetzeHipp83}
G{\"o}tze F, Hipp C. 1983  Asymptotic Expansions for Sums of Weakly Dependent
  Random Vectors. {\em Zeitschrift f{\"u}r Wahrscheinlichkeitstheorie und
  Verwandte Gebiete} \textbf{64}, 211--239.

\bibitem{FieldBook}
Field C, Ronchetti E. 1990 {\em Small Sample Asymptotics}.
{Institute of Mathematical Statistics}.

\bibitem{GoetzeHipp94}
G{\"o}tze F, Hipp C. 1994  Asymptotic distribution of statistics in time
  series. {\em Ann. Statist.} \textbf{22}, 2062--2088.

\bibitem{ait-sahalia_maximum_2002}
A\"{\i}t-Sahalia Y. 2002  Maximum {{Likelihood Estimation}} of {{Discretely
  Sampled Diffusions}}: {{A Closed}}-Form {{Approximation Approach}}. {\em
  Econometrica} \textbf{70}, 223--262.

\bibitem{HervePene10}
Herv{\'e} L, P{\`e}ne F. 2010  The {N}agaev-{G}uivarc'h method via the
  {K}eller-{L}iverani theorem. {\em Bull. Soc. Math. France} \textbf{138},
  415--489.

\bibitem{rokhlin1949fundamental}
Rokhlin VA. 1949  On the fundamental ideas of measure theory. {\em
  Matematicheskii Sbornik} \textbf{67}, 107--150.

\bibitem{wouters_multi-level_2013}
Wouters J, Lucarini V. 2013  Multi-level {Dynamical} {Systems}: {Connecting}
  the {Ruelle} {Response} {Theory} and the {Mori}-{Zwanzig} {Approach}. {\em
  Journal of Statistical Physics} \textbf{151}, 850--860.

\bibitem{wouters_disentangling_2012}
Wouters J, Lucarini V. 2012  Disentangling multi-level systems: averaging,
  correlations and memory. {\em Journal of Statistical Mechanics: Theory and
  Experiment} \textbf{2012}, P03003.

\bibitem{wouters_parameterization_2016}
Wouters J, Dolaptchiev SI, Lucarini V, Achatz U. 2016  Parameterization of
  stochastic multiscale triads. {\em Nonlin. Processes Geophys.} \textbf{23},
  435--445.

\bibitem{BaladiSmania08}
Baladi V, Smania D. 2008  Linear response formula for piecewise expanding
  unimodal maps. {\em Nonlinearity} \textbf{21}, 677.

\bibitem{BaladiSmania10}
Baladi V, Smania D. 2010  Alternative proofs of linear response for piecewise
  expanding unimodal maps. {\em Ergodic Theory and Dynamical Systems}
  \textbf{30}, 1--20.

\bibitem{Baladi14}
Baladi V. 2014  Linear Response, or else. In {\em ICM Seoul 2014, Proceedings,
  Volume {III}} pp. 525--545.

\bibitem{BaladiEtAl15}
Baladi V, Benedicks M, Schnellmann D. 2015  Whitney-{H}{\"o}lder continuity of
  the {SRB} measure for transversal families of smooth unimodal maps. {\em
  Invent. Math.} \textbf{201}, 773--844.

\bibitem{DeLimaSmania15}
de~Lima A, Smania D. 2016  Central limit theorem for the modulus of continuity
  of averages of observables on transversal families of piecewise expanding
  unimodal maps. {\em Journal of the Institute of Mathematics of Jussieu} pp.
  1--61.

\bibitem{GottwaldEtAl16}
Gottwald GA, Wormell JP, Wouters J. 2016  On spurious detection of linear
  response and misuse of the fluctuation-dissipation theorem in finite time
  series. {\em Phys. D} \textbf{331}, 89--101.

\bibitem{PavliotisStuart}
Pavliotis GA, Stuart AM. 2008 {\em Multiscale Methods: Averaging and
  Homogenization}.
New York: Springer.

\bibitem{Lorenz63}
Lorenz EN. 1963  Deterministic nonperiodic flow. {\em Journal of the
  Atmospheric Sciences} \textbf{20}, 130--141.

\bibitem{ArujoEtAl15}
Ara{\'u}jo V, Melbourne I, Varandas P. 2015  Rapid Mixing for the Lorenz
  Attractor and Statistical Limit Laws for Their Time-1 Maps. {\em
  Communications in Mathematical Physics} \textbf{340}, 901--938.

\bibitem{odejl}
 ODE.jl v2.1 commit 8d4827b93609118633478acf09a74247f47cd97e.
  \url{https://github.com/JuliaDiffEq/ODE.jl}.

\bibitem{dormand_family_1980}
Dormand J, Prince P. 1980  A Family of Embedded {{Runge}}-{{Kutta}} Formulae.
  {\em Journal of Computational and Applied Mathematics} \textbf{6}, 19--26.

\bibitem{wouters_stochastic_2018}
Wouters J, Gottwald GA. 2018  Stochastic model reduction for slow-fast systems
  with moderate time-scale separation. {\em arXiv:1804.09537 [cond-mat,
  physics:nlin]}.
arXiv: 1804.09537.

\bibitem{butterley_smooth_2007}
Butterley O, Liverani C. 2007  Smooth {Anosov} flows: correlation spectra and
  stability. {\em Journal of Modern Dynamics} \textbf{1}, 301--322.

\bibitem{Baladi2018}
Baladi V. 2017  The quest for the ultimate anisotropic {B}anach space. {\em J.
  Stat. Phys.} \textbf{166}, 525--557.

\bibitem{MR832868}
Lasota A, Mackey MC. 1985 {\em Probabilistic Properties of Deterministic
  Systems}.
Cambridge University Press, Cambridge.

\end{thebibliography}


\end{document}